\documentclass{article}
\usepackage{amsmath,bm,colortbl,graphicx,mathrsfs,afterpage,amssymb,url,cite,verbatim,amsthm,lscape}
\theoremstyle{plain}
\newtheorem{theorem}{Theorem}
\newtheorem{lemma}[theorem]{Lemma}

\theoremstyle{remark}
\newtheorem*{remark}{Remark}

\def\E{{\rm E}} 
\def\P{{\rm P}} 
\makeatletter 
\def\@biblabel#1{\hspace*{-\labelsep}}
\makeatother
\allowdisplaybreaks

\begin{document}
\title{A game-theoretic analysis of \\ baccara chemin de fer}
\author{S. N. Ethier\thanks{Partially supported by a grant from the Simons Foundation (209632).}\ \ and Carlos G\'amez\\ \noalign{\smallskip}
\textit{University of Utah and Universidad de El Salvador}}
\date{}
\maketitle

\begin{abstract}
Assuming that cards are dealt with replacement from a single deck and that each of Player and Banker sees the total of his own two-card hand but not its composition, baccara is a $2\times2^{88}$ matrix game, which was solved by Kemeny and Snell in 1957.  Assuming that cards are dealt without replacement from a $d$-deck shoe and that Banker sees the composition of his own two-card hand while Player sees only his own total, baccara is a $2\times2^{484}$ matrix game, which was solved by Downton and Lockwood in 1975 for $d=1,2,\ldots,8$.  Assuming that cards are dealt without replacement from a $d$-deck shoe and that each of Player and Banker sees the composition of his own two-card hand, baccara is a $2^5\times2^{484}$ matrix game, which is solved herein for every positive integer $d$.\medskip\par
\noindent \textit{AMS 2010 subject classification}: Primary 91A05; secondary 91A60. \vglue1.5mm\par
\noindent \textit{Key words and phrases}: baccara chemin de fer, sampling without replacement, matrix game, strict dominance, kernel, solution, infinite precision.
\end{abstract}

\section{Introduction}
The game of \textit{baccara chemin de fer} (briefly, baccara) played a key role in the development of game theory.  Bertrand's (1889, pp.~38--42) analysis of whether Player should draw or stand on a two-card total of 5 was the starting point of Borel's investigation of strategic games (Dimand and Dimand 1996, p.~132).  Borel (1924) described Bertrand's study as ``extremely incomplete'' but did not himself contribute to baccara.  It is unfortunate that Borel was unaware of Dormoy's (1873) work, which was less incomplete.  Von Neumann (1928), after proving the minimax theorem, remarked that he would analyze baccara in a subsequent paper.  But a solution of the game would have to wait until the dawn of the computer age.  Kemeny and Snell (1957), assuming that cards are dealt with replacement from a single deck and that each of Player and Banker sees the total of his own two-card hand but not its composition, found the unique solution of the resulting $2\times 2^{88}$ matrix game.  In practice, cards are dealt without replacement from a \textit{sabot}, or shoe, containing six 52-card decks.  Downton and Lockwood (1975), allowing a $d$-deck shoe dealt without replacement and assuming that Banker sees the composition of his own two-card hand while Player sees only his own total, found the unique solution of the resulting $2\times 2^{484}$ matrix game for $d=1,2,\ldots,8$.  They used an algorithm of Foster (1964).

Our aim in this paper is to solve the game without simplifying assumptions.  We allow a $d$-deck shoe dealt without replacement and allow each of Player and Banker to see the composition of his own two-card hand, making baccara a $2^5\times 2^{484}$ matrix game.  We derive optimal Player and Banker strategies and determine the value of the game, doing so for every positive integer $d$.  We too make use of Foster's (1964) algorithm.  We suspect that these optimal strategies are uniquely optimal, but we do not have a proof of uniqueness.

It will be convenient for what follows to classify the game-theoretic models of baccara in two ways.  First, we classify them according to how the cards are dealt.
\begin{itemize}
\item Model A.  Cards are dealt with replacement from a single deck.  
\item Model B.  Cards are dealt without replacement from a $d$-deck shoe.  
\end{itemize}
Second, we classify them according to the information available to Player and Banker about their own two-card hands.
\begin{itemize}
\item Model 1.  Each of Player and Banker sees the total of his own two-card hand but not its composition.
\item Model 2.  Banker sees the composition of his own two-card hand while Player sees only his own total.
\item Model 3.  Each of Player and Banker sees the composition of his own two-card hand.
\end{itemize}
(We do not consider the fourth possibility.)  Thus, Model A1 is the model of Kemeny and Snell (1957), Model B2 is the model of Downton and Lockwood (1975), and Model B3 is our primary focus here.  Model A2 was discussed by Downton and Holder (1972), but Models A3, B1, and B3 have not been considered before, as far as we know.

Like others before us, we restrict our attention to the classical parlor game of baccara chemin de fer, in contrast to the modern casino game.  (Following Deloche and Oguer 2007, we use the authentic French spelling ``baccara'' rather than the more conventional ``baccarat'' to emphasize this.)  The rules are as follows.  Denominations A, 2--9, 10, J, Q, K have values 1, 2--9, 0, 0, 0, 0, respectively.  The total of a hand, consisting of two or three cards, is the sum of the values of the cards, modulo 10.  In other words, only the final digit of the sum is used to evaluate a hand.  Two cards are dealt face down to Player and two face down to Banker, and each looks only at his own hand.  The object of the game is to have the higher total (closer to 9) at the end of play.  A two-card total of 8 or 9 is a \textit{natural}.  If either hand is a natural, the game is over and the higher total wins.  Hands of equal total result in a push (no money is exchanged).  If neither hand is a natural, Player then has the option of drawing a third card.  If he exercises this option, his third card is dealt face up.  Next, Banker, observing Player's third card, if any, has the option of drawing a third card.  This completes the game, and the higher total wins.  Winning bets on Player's hand are paid even money, with Banker, as the name suggests, playing the role of the bank.  Again, hands of equal total result in a push.  Since bystanders can bet on Player's hand, Player's strategy is restricted.  He must draw on a two-card total of 4 or less and stand on a two-card total of 6 or 7.  When his two-card total is 5, he is free to draw or stand as he chooses.  Banker, on whose hand no one can bet, has no restrictions on his strategy.

In the modern casino game, not only is Banker's strategy highly constrained but the casino collects a five percent commission on Banker wins.

In Section \ref{payoff matrix}, we show how to evaluate the payoff matrix.  We emphasize Model B3 but treat the other models as well.  In Section \ref{reduce}, we use strict dominance to reduce the payoff matrix under Model B3 to $2^5\times 2^{n_d}$, where $n_d$ depends on the number of decks $d$ and satisfies $18\le n_d\le 23$.  We get similar reductions of the other models.  To proceed further, in Section \ref{A1} we re-examine the unique solution of Kemeny and Snell (1957) under Model A1 and notice that there are multiple solutions under Models A2 and A3. In Section \ref{B1} we derive the unique solution under Model B1 for every positive integer $d$.  Model B1 is of interest because it shows the price of the ``with replacement'' assumption more clearly than do Models B2 and B3.  In Section \ref{B2} we re-derive the unique solution of Downton and Lockwood (1975) under Model B2, extending it to every positive integer $d$.  These results lead us in Section \ref{B3} to a solution under Model B3 for every positive integer $d$.  The feature of the game that allows this is that, under Model B3, the kernel is $2\times2$.  The two Banker pure strategies specified by the kernel are dependent on $d$, while the two Player pure strategies specified by the kernel are independent of $d$.  Optimality proofs are computer-assisted, with all computations carried out in infinite precision using \textit{Mathematica}.

\section{Preliminaries}\label{prelims}

The following lemma was used implicitly by Kemeny and Snell (1957), Foster (1964), and Downton and Lockwood (1975), and explicitly by Ethier (2010, p.~166).  The latter reference has a proof.

\begin{lemma}\label{Lemma1}
Let $m\ge2$ and $n\ge1$ and consider an $m\times 2^n$ matrix game of the following form.  Player I has $m$ pure strategies, labeled $0,1,\ldots,m-1$.  Player II has $2^n$ pure strategies, labeled by the subsets $T\subset\{1,2,\ldots,n\}$.  For $i=0,1,\ldots,m-1$, there exist probabilities $p_i(0)\ge0$, $p_i(1)>0$, \dots, $p_i(n)>0$ with $p_i(0)+p_i(1)+\cdots+p_i(n)=1$ together with a real number $a_i(0)$, and for $l=1,2,\ldots,n$, there exists a real $m\times 2$ matrix
\begin{equation*}
\left(
\begin{array}{cc}
a_{0,0}(l)&a_{0,1}(l)\\
a_{1,0}(l)&a_{1,1}(l)\\
\vdots & \vdots \\
a_{m-1,0}(l)&a_{m-1,1}(l)\\
\end{array}
\right).
\end{equation*}
The $m\times 2^n$ matrix game has payoff matrix with $(i,T)$ entry given by
\begin{equation}\label{payoff-form}
a_{i,T}:=p_i(0)a_i(0)+\sum_{l\in T}p_i(l)a_{i,1}(l)+\sum_{l\in T^c}p_i(l)a_{i,0}(l)
\end{equation}
for $i\in\{0,1,\ldots,m-1\}$ and $T\subset\{1,2,\ldots,n\}$.  Here $T^c:=\{1,2,\ldots,n\}-T$.

We define
\begin{eqnarray*}
T_0&:=&\{1\le l\le n: a_{i,0}(l)<a_{i,1}(l){\rm\ for\ }i=0,1,\ldots,m-1\},\nonumber\\
T_1&:=&\{1\le l\le n: a_{i,0}(l)>a_{i,1}(l){\rm\ for\ }i=0,1,\ldots,m-1\},\\
T_*&:=&\{1,2,\ldots,n\}-T_0-T_1,\nonumber
\end{eqnarray*}
and put $n_*:=|T_*|$.
Then, given $T\subset\{1,2,\ldots,n\}$, player II's pure strategy $T$ is strictly dominated unless $T_1\subset T\subset T_1\cup T_*$.  Therefore, the $m\times 2^n$ matrix game can be reduced to an $m\times 2^{n_*}$ matrix
game.
\end{lemma}

\begin{remark}
The game can be thought of as follows.  Player I chooses a pure strategy $i\in\{0,1,\ldots,m-1\}$.  Let $Z_i$ be a random variable with distribution $\P(Z_i=l)=p_i(l)$ for $l=0,1,\ldots,n$.  Given that $Z_i=0$, the game is over and player I's conditional expected gain is $a_i(0)$.  If $Z_i\in\{1,2,\ldots,n\}$, then player II observes $Z_i$ (but not $i$) and based on this information chooses a ``move'' $j\in\{0,1\}$.  Given that $Z_i=l$ and player II chooses move 1 (resp., move 0), player I's conditional expected gain is $a_{i,1}(l)$ (resp., $a_{i,0}(l)$).  Thus, player II's pure strategies can be identified with subsets $T\subset\{1,2,\ldots,n\}$, with player II choosing move 1 if $Z_i\in T$ and move 0 if $Z_i\in T^c$.  The lemma implies that, regardless of player I's strategy choice, it is optimal for player II to choose move 1 if $Z_i\in T_1$ and move 0 if $Z_i\in T_0$. 
\end{remark}

We now formalize Foster's (1964) algorithm for solving $2\times 2^n$ matrix games as described in the special case of Lemma \ref{Lemma1} in which $m=2$.   (See also Foster's discussion of Kendall and Murchland 1964.)  The method is purely algebraic, so we simplify the notation slightly by defining
$$
e_i(0):=p_i(0)a_i(0),\qquad e_{i,j}(l):=p_i(l)a_{i,j}(l)
$$
for $i,j=0,1$ and $l=1,2,\ldots,n$.

\begin{lemma}\label{Lemma2}
Let $n\ge1$ and consider a $2\times 2^n$ matrix game
of the following form.  Player I has two pure strategies, labeled $0$ and $1$.
Player II has $2^n$ pure strategies, labeled by the subsets
$T\subset\{1,2,\ldots,n\}$.  There exist real numbers $e_0(0)$ and $e_1(0)$ and,
for $l=1,2,\ldots,n$, a real $2\times 2$ matrix
\begin{equation*}
\setlength{\arraycolsep}{1.5mm} 
\left(
\begin{array}{cc}
e_{0,0}(l)&e_{0,1}(l)\\
e_{1,0}(l)&e_{1,1}(l)\\
\end{array}
\right).
\end{equation*}
The $2\times 2^n$ matrix game has payoff matrix with $(i,T)$
entry given by
\begin{equation*}
a_{i,T}:=e_i(0)+\sum_{l\in T}e_{i,1}(l)+\sum_{l\in T^c}e_{i,0}(l) 
\end{equation*}
for $i\in\{0,1\}$ and $T\subset\{1,2,\ldots,n\}$.  

We define
\begin{eqnarray*}
T_{00}&:=&\{1\le l\le n: e_{0,0}(l)<e_{0,1}(l)\text{ and }e_{1,0}(l)<e_{1,1}(l)\},\nonumber\\
T_{01}&:=&\{1\le l\le n: e_{0,0}(l)<e_{0,1}(l)\text{ and }e_{1,0}(l)>e_{1,1}(l)\},\nonumber\\
T_{10}&:=&\{1\le l\le n: e_{0,0}(l)>e_{0,1}(l)\text{ and }e_{1,0}(l)<e_{1,1}(l)\},\nonumber\\
T_{11}&:=&\{1\le l\le n: e_{0,0}(l)>e_{0,1}(l)\text{ and }e_{1,0}(l)>e_{1,1}(l)\},\nonumber
\end{eqnarray*}
and assume that $T_{00}\cup T_{01}\cup T_{10}\cup T_{11}=\{1,2,\ldots,n\}$.
If player I uses the mixed strategy $(1-p,p)$, then player II's best response is
\begin{equation}\label{T(p)}
T(p):=T_{11}\cup\{l\in T_{01}:p(l)<p\}\cup\{l\in T_{10}:p(l)>p\},
\end{equation}
where
$$
p(l):=\frac{e_{0,0}(l)-e_{0,1}(l)}{e_{0,0}(l)-e_{0,1}(l)+e_{1,1}(l)-e_{1,0}(l)}.
$$
This player I mixed strategy and player II best response leads to a player I expected gain of
\begin{eqnarray}\label{lower-env}
V(p)&:=&(1-p)e_0+p\,e_1+\sum_{\{l\in T_{01}:p(l)<p\}\cup\{l\in T_{10}:p(l)>p\}}[(1-p)e_{0,1}(l)+p\,e_{1,1}(l)]\nonumber\\
&&\quad{}+\sum_{\{l\in T_{01}:p(l)\ge p\}\cup\{l\in T_{10}:p(l)\le p\}}[(1-p)e_{0,0}(l)+p\,e_{1,0}(l)],\qquad
\end{eqnarray}
where
$$
e_i:=e_i(0)+\sum_{l\in T_{11}}e_{i,1}(l)+\sum_{l\in T_{00}}e_{i,0}(l).
$$
The function $p\mapsto V(p)$ is the lower envelope of the family of linear functions $p\mapsto (1-p)a_{0,T}+p\,a_{1,T}$, where $T$ ranges over $T_{11}\subset T\subset T_{11}\cup T_{01}\cup T_{10}$.
Therefore, the value of the game is
$$
V:=\max_{0\le p\le1}V(p)=\max\Big(V(0),V(1),\max_{l\in T_{01}\cup T_{10}}V(p(l))\Big)=V(p^*),
$$
If the last equality uniquely determines $p^*$ and if $p^*=p(l^*)$ for a unique $l^*\in T_{01}\cup T_{10}$, then player I's unique optimal strategy is $(1-p^*,p^*)$ and the two columns of the kernel are uniquely specified as $T(p^*)$ and $T(p^*)\cup\{l^*\}$.  Their unique optimal mixture $(1-q^*,q^*)$ is obtained by solving the $2\times2$ kernel.
\end{lemma}

\begin{proof}
Notice that $l$ belongs to $T(p)$ if $l\in T_{11}$ or if both $l\in T_{01}\cup T_{10}$ and
$$
(1-p)e_{0,1}(l)+p\,e_{1,1}(l)<(1-p)e_{0,0}(l)+p\,e_{1,0}(l),
$$
implying (\ref{T(p)}).  The function (\ref{lower-env}) is continuous and piecewise linear, hence is maximized at 0, 1, or one of the points $p(l)$ at which its slope changes.  The remaining conclusions of the lemma follow easily.
\end{proof}

\section{Evaluation of the payoff matrix}\label{payoff matrix}

We begin by considering the game under Model B3.  Let $X_1\le X_2$ be the values of the two cards dealt to Player and $Y_1\le Y_2$ the values of the two cards dealt to Banker.  Define the function $M:\{0,1,\ldots\}\mapsto\{0,1,\ldots,9\}$ by $M(i):=\text{Mod}(i,10)$.  Then $X:=M(X_1+X_2)$ is Player's two-card total and $Y:=M(Y_1+Y_2)$ is Banker's two-card total.  On the event $\{X\le7,\;Y\le7\}$, let $X_3$  denote the value of Player's third card if he draws, and let $X_3:=\varnothing$ if he stands.  Similarly, let $Y_3$ denote the value of Banker's third card if he draws, and let $Y_3:=\varnothing$ if he stands.  

As the rules specify, Player's pure strategies can be indexed by the sets $S$ satisfying
\begin{eqnarray}\label{S constraints}
&&\{(i_1,i_2):0\le i_1\le i_2\le9: M(i_1+i_2)\le4\}\nonumber\\
&&\qquad{}\subset S\subset\{(i_1,i_2):0\le i_1\le i_2\le9: M(i_1+i_2)\le5\}.
\end{eqnarray}
Assuming $X\le 7$ and $Y\le7$, Player draws if $(X_1,X_2)\in S$ and stands otherwise.  Since the set of pairs $(i_1,i_2)$ satisfying $0\le i_1\le i_2\le9$ and $M(i_1+i_2)=5$ contains $(0,5)$, $(1,4)$, $(2,3)$, $(6,9)$, and $(7,8)$, it follows that Player has $2^5$ pure strategies.

On the other hand, Banker's pure strategies can be indexed by the sets $T$ satisfying
\begin{equation}\label{T constraints}
T\subset\{(j_1,j_2):0\le j_1\le j_2\le9,\, M(j_1+j_2)\le7\}\times\{0,1,2,\ldots,9,\varnothing\}.
\end{equation}
Assuming $X\le 7$ and $Y\le7$, Banker draws if $(Y_1,Y_2,X_3)\in T$ and stands otherwise.  Since there are 44 pairs $(j_1,j_2)$ satisfying $0\le j_1\le j_2\le 9$ and $M(j_1+j_2)\le7$, and since $44\times11=484$, it follows that Banker has $2^{484}$ pure strategies.

Thus, baccara is a $2^5\times 2^{484}$ matrix game.  Let us denote by $G_{S,T}$ Player's profit from a one-unit bet when he adopts pure strategy $S$ and Banker adopts pure strategy $T$, so that $a_{S,T}:=\E[G_{S,T}]$ is the $(S,T)$ entry in the payoff matrix.
Then
\begin{eqnarray}\label{a_{S,T}prelim}
a_{S,T}&=&\E[G_{S,T}]\nonumber\\
&=&\P(X\in\{8,9\},\;X>Y)-\P(Y\in\{8,9\},\;Y>X)\nonumber\\
&&\quad{}+\E[G_{S,T}\,1_{\{X\le7,\;Y\le7\}}]\nonumber\\
&=&\E[G_{S,T}\,1_{\{X\le7,\;Y\le7\}}]\nonumber\\
&=&\sum_{M(j_1+j_2)\le7}\sum_{k=0}^9\P((X_1,X_2)\in S,\; (Y_1,Y_2)=(j_1,j_2),\; X_3=k)\nonumber\\
&&\qquad\quad{}\cdot\E[G_{S,T}\mid (X_1,X_2)\in S,\; (Y_1,Y_2)=(j_1,j_2),\; X_3=k]\nonumber\\
&&{}+\sum_{M(j_1+j_2)\le7}\P((X_1,X_2)\in S^c,\; (Y_1,Y_2)=(j_1,j_2),\; X_3=\varnothing)\nonumber\\
&&\qquad\quad{}\cdot\E[G_{S,T}\mid (X_1,X_2)\in S^c,\; (Y_1,Y_2)=(j_1,j_2),\; X_3=\varnothing],
\end{eqnarray}
where
$S^c:=\{(i_1,i_2):0\le i_1\le i_2\le9,\, M(i_1+i_2)\le7\}-S$.

Let us now define, for $S$ and $T$, for $(j_1,j_2)$ satisfying $0\le j_1\le j_2\le9$ and $M(j_1+j_2)\le7$, and for $k\in\{0,1,\ldots,9\}$,
\begin{equation}\label{a_{S,l}}
\begin{split}
a_{S,l}(j_1,j_2,k)&:=\E[G_{S,T}\mid (X_1,X_2)\in S,\; (Y_1,Y_2)=(j_1,j_2),\; X_3=k],\\
a_{S,l}(j_1,j_2,\varnothing)&:=\E[G_{S,T}\mid (X_1,X_2)\in S^c,\; (Y_1,Y_2)=(j_1,j_2),\;X_3=\varnothing],
\end{split}
\end{equation}
where $l=1$ if $(j_1,j_2,k)$ (resp., $(j_1,j_2,\varnothing)$) belongs to $T$; and $l=0$
if $(j_1,j_2,k)$ (resp., $(j_1,j_2,\varnothing)$) belongs to
$T^c$ (the complement of $T$ relative to (\ref{T constraints})).  Defining also
\begin{equation*}
\begin{split}
p_S(j_1,j_2,k)&:=\P((X_1,X_2)\in S,\; (Y_1,Y_2)=(j_1,j_2),\; X_3=k),\\
p_S(j_1,j_2,\varnothing)&:=\P((X_1,X_2)\in S^c,\; (Y_1,Y_2)=(j_1,j_2),\; X_3=\varnothing),
\end{split}
\end{equation*}
we have, from (\ref{a_{S,T}prelim}),
\begin{eqnarray}\label{a_{S,T}}
a_{S,T}&=&\sum_{(j_1,j_2,k)\in T\text{ with }k\ne\varnothing}p_S(j_1,j_2,k)a_{S,1}(j_1,j_2,k)\nonumber\\
&&\qquad{}+\sum_{(j_1,j_2,k)\in T^c\text{ with }k\ne\varnothing}p_S(j_1,j_2,k)a_{S,0}(j_1,j_2,k)\nonumber\\
&&\qquad{}+\sum_{(j_1,j_2,\varnothing)\in T}p_S(j_1,j_2,\varnothing)a_{S,1}(j_1,j_2,\varnothing)\nonumber\\
&&\qquad{}+\sum_{(j_1,j_2,\varnothing)\in T^c}p_S(j_1,j_2,\varnothing)a_{S,0}(j_1,j_2,\varnothing).
\end{eqnarray}

To evaluate the conditional expectations in (\ref{a_{S,l}}), we condition on $(X_1,X_2)$:
\begin{eqnarray}\label{a_{S,l}draw}
&&a_{S,l}(j_1,j_2,k)\nonumber\\
&&{}=\sum_{(i_1,i_2)\in S}\P((X_1,X_2)=(i_1,i_2)\mid (X_1,X_2)\in S,\;(Y_1,Y_2)=(j_1,j_2),\; X_3=k)\nonumber\\
&&\qquad\qquad{}\cdot\E[G_{S,T}\mid (X_1,X_2)=(i_1,i_2),\; (Y_1,Y_2)=(j_1,j_2),\; X_3=k]
\end{eqnarray}
if $k\ne\varnothing$ and
\begin{eqnarray}\label{a_{S,l}stand}
&&a_{S,l}(j_1,j_2,\varnothing)\nonumber\\
&&{}=\sum_{(i_1,i_2)\in S^c}\P((X_1,X_2)=(i_1,i_2)\mid (X_1,X_2)\in S^c,\;(Y_1,Y_2)=(j_1,j_2),\; X_3=\varnothing)\nonumber\\
&&\qquad\qquad{}\cdot\E[G_{S,T}\mid (X_1,X_2)=(i_1,i_2),\; (Y_1,Y_2)=(j_1,j_2),\; X_3=\varnothing].
\end{eqnarray}
To evaluate the conditional expectations in (\ref{a_{S,l}draw}) and (\ref{a_{S,l}stand}), there are four cases to consider:

Case 1.  $(i_1,i_2)\in S$, $(j_1,j_2,k)\in T$ with $k\ne\varnothing$ (both Player and Banker draw).  Here, for $d$ decks,
\begin{eqnarray}\label{case 1, d finite}
&&\E[G_{S,T}\mid (X_1,X_2)=(i_1,i_2),\; (Y_1,Y_2)=(j_1,j_2),\; X_3=k]\nonumber\\
&&\quad{}=\sum_{l=0}^9\frac{4d(1+3\delta_{l,0})-\delta_{l,i_1}-\delta_{l,i_2}-\delta_{l,j_1}-\delta_{l,j_2}-\delta_{l,k}}{52d-5}\nonumber\\
&&\qquad\qquad\qquad{}\cdot{\rm sgn}(M(i_1+i_2+k)-M(j_1+j_2+l)),
\end{eqnarray}
which becomes, as $d\to\infty$, 
\begin{equation}\label{case 1, d infinite}
=\sum_{l=0}^9\frac{1+3\delta_{l,0}}{13}\,{\rm sgn}(M(i_1+i_2+k)-M(j_1+j_2+l)).
\end{equation}

Case 2.  $(i_1,i_2)\in S$, $(j_1,j_2,k)\in T^c$ with $k\ne\varnothing$ (Player draws, Banker stands).  Regardless of the number of decks,
\begin{eqnarray*}
&&\E[G_{S,T}\mid (X_1,X_2)=(i_1,i_2),\; (Y_1,Y_2)=(j_1,j_2),\; X_3=k]\nonumber\\
&&\quad{}={\rm sgn}(M(i_1+i_2+k)-M(j_1+j_2)).
\end{eqnarray*}

Case 3.  $(i_1,i_2)\in S^c$, $(j_1,j_2,\varnothing)\in T$ (Player stands, Banker draws).  For $d$ decks,
\begin{eqnarray}\label{case 3, d finite}
&&\E[G_{S,T}\mid (X_1,X_2)=(i_1,i_2),\; (Y_1,Y_2)=(j_1,j_2),\; X_3=\varnothing]\nonumber\\
&&\quad{}=\sum_{l=0}^9\frac{4d(1+3\delta_{l,0})-\delta_{l,i_1}-\delta_{l,i_2}-\delta_{l,j_1}-\delta_{l,j_2}}{52d-4}\nonumber\\
&&\qquad\qquad\qquad{}\cdot{\rm sgn}(M(i_1+i_2)-M(j_1+j_2+l)),
\end{eqnarray}
which becomes, as $d\to\infty$, 
\begin{equation}\label{case 3, d infinite}
=\sum_{l=0}^9\frac{1+3\delta_{l,0}}{13}\,{\rm sgn}(M(i_1+i_2)-M(j_1+j_2+l)).
\end{equation}

Case 4.  $(i_1,i_2)\in S^c$, $(j_1,j_2,\varnothing)\in T^c$ (both Player and Banker stand).  Regardless of the number of decks,
\begin{eqnarray*}
&&\E[G_{S,T}\mid (X_1,X_2)=(i_1,i_2),\; (Y_1,Y_2)=(j_1,j_2),\; X_3=\varnothing]\nonumber\\
&&\quad{}={\rm sgn}(M(i_1+i_2)-M(j_1+j_2)).
\end{eqnarray*}

Finally, to evaluate
\begin{equation*}
\P((X_1,X_2)=(i_1,i_2)\mid (X_1,X_2)\in S,\; (Y_1,Y_2)=(j_1,j_2),\; X_3=k),
\end{equation*}
we begin with a full $d$-deck shoe except for three cards, one $j_1$, one $j_2$, and one $k$, removed.  It will comprise $m_0$ 0s, $m_1$ 1s, \dots, and $m_9$ 9s, where
$$
m_r:=4d(1+3\delta_{r,0})-\delta_{r,j_1}-\delta_{r,j_2}-\delta_{r,k},\quad r=0,1,\ldots,9.
$$
The number of equally likely two-card hands is $\binom{52d-3}{2}$, and the number of those that belong to $S$ is
$$
m:=\sum_{(u_1,u_2)\in S}\frac{m_{u_1}(m_{u_2}-\delta_{u_2,u_1})}{1+\delta_{u_1,u_2}}.
$$
Then
\begin{eqnarray}\label{wo-repl1}
&&\P((X_1,X_2)=(i_1,i_2)\mid (X_1,X_2)\in S,\; (Y_1,Y_2)=(j_1,j_2),\; X_3=k)\nonumber\\
&&\quad{}=\frac{m_{i_1}(m_{i_2}-\delta_{i_2,i_1})}{1+\delta_{i_1,i_2}}\,\frac{1_S((i_1,i_2))}{m}
\end{eqnarray}
and
\begin{eqnarray}\label{wo-repl2}
p_S(j_1,j_2,k)&=&\P((Y_1,Y_2)=(j_1,j_2),\; X_3=k)\nonumber\\
&&\quad{}\cdot\P((X_1,X_2)\in S\mid(Y_1,Y_2)=(j_1,j_2),\; X_3=k)\nonumber\\
&=&\frac{(2-\delta_{j_1,j_2})4d(1+3\delta_{j_1,0})[4d(1+3\delta_{j_2,0})-\delta_{j_2,j_1}]}{(52d)_2}\nonumber\\
&&\quad{}\cdot\frac{4d(1+3\delta_{k,0})-\delta_{k,j_1}-\delta_{k,j_2}}{52d-2}\,\frac{m}{\binom{52d-3}{2}}.
\end{eqnarray}

Also, to evaluate
\begin{equation*}
\P((X_1,X_2)=(i_1,i_2)\mid (X_1,X_2)\in S^c,\; (Y_1,Y_2)=(j_1,j_2),\; X_3=\varnothing),
\end{equation*}
we begin with a full $d$-deck shoe except for two cards, one $j_1$ and one $j_2$, removed.  It will comprise $m_0'$ 0s, $m_1'$ 1s, \dots, and $m_9'$ 9s, where
$$
m_r':=4d(1+3\delta_{r,0})-\delta_{r,j_1}-\delta_{r,j_2},\quad r=0,1,\ldots,9.
$$
The number of equally likely two-card hands is $\binom{52d-2}{2}$, and the number of those that belong to $S^c$ is
$$
m':=\sum_{(u_1,u_2)\in S^c}\frac{m_{u_1}'(m_{u_2}'-\delta_{u_2,u_1})}{1+\delta_{u_1,u_2}}.
$$
Then
\begin{eqnarray}\label{wo-repl3}
&&\P((X_1,X_2)=(i_1,i_2)\mid (X_1,X_2)\in S^c,\; (Y_1,Y_2)=(j_1,j_2),\; X_3=\varnothing)\nonumber\\
&&\quad{}=\frac{m_{i_1}'(m_{i_2}'-\delta_{i_2,i_1})}{1+\delta_{i_1,i_2}}\,\frac{1_{S^c}((i_1,i_2))}{m'}
\end{eqnarray}
and
\begin{eqnarray}\label{wo-repl4}
p_S(j_1,j_2,\varnothing)&=&\P((Y_1,Y_2)=(j_1,j_2),\; X_3=\varnothing)\nonumber\\
&&\quad{}\cdot\P((X_1,X_2)\in S^c\mid(Y_1,Y_2)=(j_1,j_2),\; X_3=\varnothing)\\
&=&\frac{(2-\delta_{j_1,j_2})4d(1+3\delta_{j_1,0})[4d(1+3\delta_{j_2,0})-\delta_{j_2,j_1}]}{(52d)_2}\,
\frac{m'}{\binom{52d-2}{2}}.\nonumber
\end{eqnarray}
This suffices to complete the evaluation of (\ref{a_{S,l}draw}) and (\ref{a_{S,l}stand}) when cards are dealt without replacement from a $d$-deck shoe.

The assumption that cards are dealt \textit{with} replacement from a single deck can be modeled by letting $d\to\infty$ in the assumption that cards are dealt without replacement from a $d$-deck shoe.  The formulas are simpler in this case:
\begin{eqnarray}\label{with-repl1}
&&\P((X_1,X_2)=(i_1,i_2)\mid (X_1,X_2)\in S,\; (Y_1,Y_2)=(j_1,j_2),\; X_3=k)\nonumber\\
&&\quad{}=\frac{(2-\delta_{i_1,i_2})(1+3\delta_{i_1,0})(1+3\delta_{i_2,0})1_S((i_1,i_2))}{89+8\,|S\cap\{(0,5)\}|+2\,|S\cap\{(1,4),(2,3),(6,9),(7,8)\}|},
\end{eqnarray}
\begin{eqnarray}\label{with-repl2}
&&\!\!\!\!\!\!\!\!\!\!\!\!p_S(j_1,j_2,k)\nonumber\\
&=&\frac{(2-\delta_{j_1,j_2})(1+3\delta_{j_1,0})(1+3\delta_{j_2,0})}{(13)^2}\,\frac{1+3\delta_{k,0}}{13}\nonumber\\
&&\quad{}\cdot \frac{89+8\,|S\cap\{(0,5)\}|+2\,|S\cap\{(1,4),(2,3),(6,9),(7,8)\}|}{(13)^2},
\end{eqnarray}
\begin{eqnarray}\label{with-repl3}
&&\P((X_1,X_2)=(i_1,i_2)\mid (X_1,X_2)\in S^c,\; (Y_1,Y_2)=(j_1,j_2),\; X_3=\varnothing)\nonumber\\
&&\quad{}=\frac{(2-\delta_{i_1,i_2})(1+3\delta_{i_1,0})(1+3\delta_{i_2,0})1_{S^c}((i_1,i_2))}{32+8\,|S^c\cap\{(0,5)\}|+2\,|S^c\cap\{(1,4),(2,3),(6,9),(7,8)\}|},
\end{eqnarray}
\begin{eqnarray}\label{with-repl4}
&&\!\!\!\!\!\!\!\!\!\!\!\!p_S(j_1,j_2,\varnothing)\nonumber\\
&=&\frac{(2-\delta_{j_1,j_2})(1+3\delta_{j_1,0})(1+3\delta_{j_2,0})}{(13)^2}\nonumber\\
&&\quad{}\cdot\frac{32+8\,|S^c\cap\{(0,5)\}|+2\,|S^c\cap\{(1,4),(2,3),(6,9),(7,8)\}|}{(13)^2}.
\end{eqnarray}
Here 89 comes from $25+16+16+16+16$, where the summands correspond to totals $0,1,2,3,4$; 32 is $16+16$, corresponding to totals 6 and 7. 

In summary, we can evaluate (\ref{a_{S,T}}) under Model B3 or A3.  Restricting $S$ to the two extremes in (\ref{S constraints}), we obtain (\ref{a_{S,T}}) under Model B2 or A2 as a special case.  Finally, as for Models B1 and A1, we can derive the analogue of (\ref{a_{S,T}}) from results already obtained.  Specifically,
\begin{eqnarray*}
a_{S^\circ,T^\circ}&=&\sum_{(j,k)\in T^\circ\text{ with }k\ne\varnothing}p_{S^\circ}(j,k)a_{S^\circ,1}(j,k)\\
&&\qquad{}+\sum_{(j,k)\in (T^\circ)^c\text{ with }k\ne\varnothing}p_{S^\circ}(j,k)a_{S^\circ,0}(j,k)\\
&&\qquad{}+\sum_{(j,\varnothing)\in T^\circ}p_{S^\circ}(j,\varnothing)a_{S^\circ,1}(j,\varnothing)\\
&&\qquad{}+\sum_{(j,\varnothing)\in (T^\circ)^c}p_{S^\circ}(j,\varnothing)a_{S^\circ,0}(j,\varnothing),
\end{eqnarray*}
where
\begin{eqnarray*}
a_{S^\circ,l}(j,k)&:=&P(G_{S^\circ,T^\circ}\mid X\in S^\circ,\,Y=j,\,X_3=k)\\
&\;=&\sum_{M(j_1+j_2)=j}p_S(j_1,j_2,k)a_{S,l}(j_1,j_2,k)/p_{S^\circ}(j,k),
\end{eqnarray*}
\begin{eqnarray*}
a_{S^\circ,l}(j,\varnothing)&:=&P(G_{S^\circ,T^\circ}\mid X\in (S^\circ)^c,\,Y=j,\,X_3=\varnothing)\\
&\;=&\sum_{M(j_1+j_2)=j}p_S(j_1,j_2,\varnothing)a_{S,l}(j_1,j_2,\varnothing)/p_{S^\circ}(j,\varnothing),
\end{eqnarray*}
\begin{eqnarray*}
p_{S^\circ}(j,k)&=&P(X\in S^\circ,\,Y=j,\,X_3=k)=\sum_{M(j_1+j_2)=j}p_S(j_1,j_2,k),
\end{eqnarray*}
and
\begin{eqnarray*}
p_{S^\circ}(j,\varnothing)&=&P(X\in (S^\circ)^c,\,Y=j,\,X_3=\varnothing)=\sum_{M(j_1+j_2)=j}p_S(j_1,j_2,\varnothing);
\end{eqnarray*}
here $S^\circ=\{0,1,2,3,4\}$ or $\{0,1,2,3,4,5\}$, $(S^\circ)^c=\{0,1,\ldots,7\}-S^\circ$, and $T^\circ\subset\{0,1,\ldots,7\}\times\{0,1,\ldots,9,\varnothing\}$, while $S$ and $T$ are the corresponding subsets of $\{(i_1,i_2):0\le i_1\le i_2\le9,\,M(i_1+i_2)\le7\}$ and $\{(j_1,j_2):0\le j_1\le j_2\le9,\,M(j_1+j_2)\le7\}\times\{0,1,\ldots,9,\varnothing\}$, respectively.

\section{Banker's strictly dominated pure strategies}\label{reduce}

Our next step is to show that Lemma~\ref{Lemma1} applies (with Player and Banker playing the roles of player I and player II, respectively), allowing us to reduce the game to a more manageable size.  The payoff matrix (\ref{a_{S,T}}) has the form (\ref{payoff-form}) with $m=32$, $n=484$, $p_i(0)=\P(X\in\{8,9\}{\rm\ or\ }Y\in\{8,9\})$, and $a_i(0)=0$.  It remains to evaluate $T_0$, $T_1$, and $T_*$ of the lemma.  

Results are summarized in Table~\ref{reduction B3}.  $T_1$ (resp., $T_0$) is the set of triples $(j_1,j_2,k)$ for which $a_{S,1}(j_1,j_2,k)<a_{S,0}(j_1,j_2,k)$ (resp., $>$) for each of Player's $2^5$ pure strategies $S$, indicated by a D (resp., S) in the corresponding entry of the table.  $T_*$ is the remaining set of triples $(j_1,j_2,k)$, indicated by a $*$ in the corresponding entry of the table.  Of particular interest is $n_d:=|T_*|$.

\begin{table}[htb]
\caption{\label{reduction B3}Banker's optimal move (preliminary version) under Model B2 or B3 with $d=6$, indicated by D (draw) or S (stand), except in the $n_6=18$ cases indicated by $*$ in which it depends on Player's strategy.  Adjustments to the table for other positive integers $d$ are specified by footnotes.   Under Model A2 or A3, results are the same as those under Model B2 or B3 with $d\ge11$.\medskip}
\catcode`@=\active\def@{\phantom{0}}
\catcode`#=\active\def#{\phantom{$^0$}}
\begin{center}
\begin{tabular}{ccccccccccccccc}
\hline
\noalign{\smallskip}
\multicolumn{2}{c}{Banker's}&&\multicolumn{11}{c}{Player's third card ($\varnothing$ if Player stands)}\\
\multicolumn{2}{c}{two-card}&&&&&&&&&&&&\\
total & hand && 0# & 1# & 2# & 3# & 4# & 5# & 6# & 7# & 8# & 9# & $\varnothing$ \\
\noalign{\smallskip} \hline
\noalign{\smallskip}
$0,1,2$&    && \cellcolor[gray]{0.85}D# & \cellcolor[gray]{0.85}D# & \cellcolor[gray]{0.85}D# & \cellcolor[gray]{0.85}D# & \cellcolor[gray]{0.85}D# & \cellcolor[gray]{0.85}D# & \cellcolor[gray]{0.85}D# & \cellcolor[gray]{0.85}D# & \cellcolor[gray]{0.85}D# & \cellcolor[gray]{0.85}D# & \cellcolor[gray]{0.85}D\\
\noalign{\smallskip} \hline
\noalign{\smallskip}
3&$(0,3)$&&    \cellcolor[gray]{0.85}D# & \cellcolor[gray]{0.85}D# & \cellcolor[gray]{0.85}D# & \cellcolor[gray]{0.85}D# & \cellcolor[gray]{0.85}D# & \cellcolor[gray]{0.85}D# & \cellcolor[gray]{0.85}D# & \cellcolor[gray]{0.85}D# & S$^2$ & $*$# & \cellcolor[gray]{0.85}D\\
3&$(1,2)$&&    \cellcolor[gray]{0.85}D# & \cellcolor[gray]{0.85}D# & \cellcolor[gray]{0.85}D# & \cellcolor[gray]{0.85}D# & \cellcolor[gray]{0.85}D# & \cellcolor[gray]{0.85}D# & \cellcolor[gray]{0.85}D# & \cellcolor[gray]{0.85}D# & S$^1$ & $*$# & \cellcolor[gray]{0.85}D\\
3&$(4,9)$&&    \cellcolor[gray]{0.85}D# & \cellcolor[gray]{0.85}D# & \cellcolor[gray]{0.85}D# & \cellcolor[gray]{0.85}D# & \cellcolor[gray]{0.85}D# & \cellcolor[gray]{0.85}D# & \cellcolor[gray]{0.85}D# & \cellcolor[gray]{0.85}D# & S$^5$ & $*$# & \cellcolor[gray]{0.85}D\\
3&$(5,8)$&&    \cellcolor[gray]{0.85}D# & \cellcolor[gray]{0.85}D# & \cellcolor[gray]{0.85}D# & \cellcolor[gray]{0.85}D# & \cellcolor[gray]{0.85}D# & \cellcolor[gray]{0.85}D# & \cellcolor[gray]{0.85}D# & \cellcolor[gray]{0.85}D# & S$^3$ & $*$# & \cellcolor[gray]{0.85}D\\
3&$(6,7)$&&    \cellcolor[gray]{0.85}D# & \cellcolor[gray]{0.85}D# & \cellcolor[gray]{0.85}D# & \cellcolor[gray]{0.85}D# & \cellcolor[gray]{0.85}D# & \cellcolor[gray]{0.85}D# & \cellcolor[gray]{0.85}D# & \cellcolor[gray]{0.85}D# & S$^3$ & $*$# & \cellcolor[gray]{0.85}D\\
\noalign{\smallskip} \hline
\noalign{\smallskip}
4&$(0,4)$&&    S# &  S$^8$  & \cellcolor[gray]{0.85}D$^1$ & \cellcolor[gray]{0.85}D# & \cellcolor[gray]{0.85}D# & \cellcolor[gray]{0.85}D# & \cellcolor[gray]{0.85}D# & \cellcolor[gray]{0.85}D# & S# & S# & \cellcolor[gray]{0.85}D\\
4&$(1,3)$&&    S# &  S$^7$  & \cellcolor[gray]{0.85}D$^1$ & \cellcolor[gray]{0.85}D# & \cellcolor[gray]{0.85}D# & \cellcolor[gray]{0.85}D# & \cellcolor[gray]{0.85}D# & \cellcolor[gray]{0.85}D# & S# & S# & \cellcolor[gray]{0.85}D\\
4&$(2,2)$&&    S# &  $*^3$  & \cellcolor[gray]{0.85}D$^1$ & \cellcolor[gray]{0.85}D# & \cellcolor[gray]{0.85}D# & \cellcolor[gray]{0.85}D# & \cellcolor[gray]{0.85}D# & \cellcolor[gray]{0.85}D# & S# & S# & \cellcolor[gray]{0.85}D\\
4&$(5,9)$&&    S# &  S$^7$  & \cellcolor[gray]{0.85}D$^1$ & \cellcolor[gray]{0.85}D# & \cellcolor[gray]{0.85}D# & \cellcolor[gray]{0.85}D# & \cellcolor[gray]{0.85}D# & \cellcolor[gray]{0.85}D# & S# & S# & \cellcolor[gray]{0.85}D\\
4&$(6,8)$&&    S# &  $*$#    & \cellcolor[gray]{0.85}D#   & \cellcolor[gray]{0.85}D# & \cellcolor[gray]{0.85}D# & \cellcolor[gray]{0.85}D# & \cellcolor[gray]{0.85}D# & \cellcolor[gray]{0.85}D# & S# & S# & \cellcolor[gray]{0.85}D\\
4&$(7,7)$&&    S# &  $*$#    & \cellcolor[gray]{0.85}D#   & \cellcolor[gray]{0.85}D# & \cellcolor[gray]{0.85}D# & \cellcolor[gray]{0.85}D# & \cellcolor[gray]{0.85}D# & \cellcolor[gray]{0.85}D# & S# & S# & \cellcolor[gray]{0.85}D\\
\noalign{\smallskip} \hline
\noalign{\smallskip}
5&$(0,5)$&&    S# & S# & S# & S# & $*^1$ & \cellcolor[gray]{0.85}D# & \cellcolor[gray]{0.85}D# & \cellcolor[gray]{0.85}D# & S# & S# & \cellcolor[gray]{0.85}D\\
5&$(1,4)$&&    S# & S# & S# & S# & S$^7$ & \cellcolor[gray]{0.85}D# & \cellcolor[gray]{0.85}D# & \cellcolor[gray]{0.85}D# & S# & S# & \cellcolor[gray]{0.85}D\\
5&$(2,3)$&&    S# & S# & S# & S# & S$^8$ & \cellcolor[gray]{0.85}D# & \cellcolor[gray]{0.85}D# & \cellcolor[gray]{0.85}D# & S# & S# & \cellcolor[gray]{0.85}D\\
5&$(6,9)$&&    S# & S# & S# & S# & $*^1$ & \cellcolor[gray]{0.85}D# & \cellcolor[gray]{0.85}D# & \cellcolor[gray]{0.85}D# & S# & S# & \cellcolor[gray]{0.85}D\\
5&$(7,8)$&&    S# & S# & S# & S# & $*^1$ & \cellcolor[gray]{0.85}D# & \cellcolor[gray]{0.85}D# & \cellcolor[gray]{0.85}D# & S# & S# & \cellcolor[gray]{0.85}D\\
\noalign{\smallskip} \hline
\noalign{\smallskip}
6&$(0,6)$&&    S# & S# & S# & S# & S# & S# &  \cellcolor[gray]{0.85}D#  & \cellcolor[gray]{0.85}D# & S# & S# & $*$\\
6&$(1,5)$&&    S# & S# & S# & S# & S# & S# &  \cellcolor[gray]{0.85}D#  & \cellcolor[gray]{0.85}D# & S# & S# & $*$\\
6&$(2,4)$&&    S# & S# & S# & S# & S# & S# &  \cellcolor[gray]{0.85}D#  & \cellcolor[gray]{0.85}D# & S# & S# & $*$\\
6&$(3,3)$&&    S# & S# & S# & S# & S# & S# & $*^{11}$ & \cellcolor[gray]{0.85}D# & S# & S# & $*$\\
6&$(7,9)$&&    S# & S# & S# & S# & S# & S# &  \cellcolor[gray]{0.85}D#  & \cellcolor[gray]{0.85}D# & S# & S# & $*$\\
6&$(8,8)$&&    S# & S# & S# & S# & S# & S# &  \cellcolor[gray]{0.85}D#  & \cellcolor[gray]{0.85}D# & S# & S# & $*$\\
\noalign{\smallskip} \hline
\noalign{\smallskip}
7&     &&       S# & S# & S# & S# & S# & S# & S# & S# & S# & S# & S\\
\noalign{\smallskip}
\hline
\noalign{\smallskip}
\multicolumn{14}{l}{$^1$Replace S by $*$ and D by $*$ and $*$ by S if $d=1$.}\\
\multicolumn{8}{l}{$^2$Replace S by $*$ if $d\le2$.}&\multicolumn{6}{l}{$^7$Replace S by $*$ if $d\ge7$.}\\
\multicolumn{8}{l}{$^3$Replace S by $*$ and $*$ by S if $d\le3$.}&\multicolumn{6}{l}{$^8$Replace S by $*$ if $d\ge8$.}\\
\multicolumn{8}{l}{$^5$Replace S by $*$ if $d\le5$.}&\multicolumn{6}{l}{$^{11}$Replace $*$ by D if $d\ge11$.}\\
\end{tabular}
\end{center}
\end{table}
\afterpage{\clearpage}

\begin{theorem}\label{thm reduce B3}
\emph{(a)} Under Model B3 with the number of decks being a positive integer $d$, Lemma \ref{Lemma1} applies.  The sets $T_0$, $T_1$, and $T_*$ of the lemma can be inferred from Table \ref{reduction B3}, with entries \emph{S}, \emph{D}, and $*$ located at elements of $T_0$, $T_1$, and $T_*$, respectively.  In particular, for $d=1,2,\ldots,10$, $n_d=23$, $21$, $20$, $19$, $19$, $18$, $21$, $23$, $23$, $23$, respectively, and, if $d\ge11$, $n_d=22$.

\emph{(b)} Exactly the same conclusions hold under Model B2.

\emph{(c)} Lemma \ref{Lemma1} applies under Models A2 and A3, with the results those under Models B2 and B3 with $d\ge11$.
\end{theorem}

\begin{proof}
It will occasionally be convenient to label the $2^5$ choices of $S$ by the integers 0 to 31.  Strategy 0 (resp., strategy 31) denotes Player's pure strategy of standing (resp., drawing) on a two-card total of 5, regardless of its composition.  More generally, strategy $u\in\{0,1,\ldots,31\}$ is specified by the 5-bit binary representation of $u$.  For example, strategy 19 (binary 10011) corresponds to drawing on $(0,5)$, standing on $(1,4)$ and $(2,3)$, and drawing on $(6,9)$ and $(7,8)$.

Table \ref{reduction B3} is identical under Models B2 and B3 because, defining 
$$
b_u(j_1,j_2,k):=a_{u,1}(j_1,j_2,k)-a_{u,0}(j_1,j_2,k) 
$$
for $u=0,1,\ldots,31$, $0\le j_1\le j_1\le 9$ with $M(j_1+j_2)\le7$, and $k=0,1,\ldots,9,\varnothing$,
we have, with a few exceptions,
\begin{equation*}
b_u(j_1,j_2,k)\in[b_0(j_1,j_2,k)\wedge b_{31}(j_1,j_2,k),b_0(j_1,j_2,k)\vee b_{31}(j_1,j_2,k)]
\end{equation*}
for $u=1,2,\ldots,30$.  The exceptions occur only when $d=1$ and only when $(j_1,j_2,k)=(0,0,9)$, $(5,5,9)$, or $(5,6,0)$.  (See Section \ref{B2} for an explanation of why this is to be expected.)

It follows that, if $b_0(j_1,j_2,k)<0$ and $b_{31}(j_1,j_2,k)<0$, then the $(j_1,j_2,k)$ entry in Table \ref{reduction B3} is D (draw); in 301 of the 484 entries, this property holds for every $d\ge1$.  If $b_0(j_1,j_2,k)>0$ and $b_{31}(j_1,j_2,k)>0$, then the $(j_1,j_2,k)$ entry in Table \ref{reduction B3} is S (stand); in 151 of the 484 entries, this property holds for every $d\ge1$.  If $b_0(j_1,j_2,k)>0>b_{31}(j_1,j_2,k)$ or $b_0(j_1,j_2,k)<0<b_{31}(j_1,j_2,k)$, then the $(j_1,j_2,k)$ entry in Table \ref{reduction B3} is $*$;  in 13 of the 484 entries, this property holds for every $d\ge1$.  This accounts for all but 19 entries in Table \ref{reduction B3}, those marked with footnotes, in which the sign of $b_0(j_1,j_2,k)$ or $b_{31}(j_1,j_2,k)$ depends on $d$.

For example, one of the 19 is $(3,3,6)$.  Indeed,
$$
b_{31}(3,3,6)=-\frac{2 (80\,d^3-832\,d^2+135\,d-2)}{(52\,d-5) (840\,d^2-114\,d+1)}>0\text{ (resp., $<0$)}
$$
for all $d\le10$ (resp., $d\ge11$) and
$$
b_0(3,3,6)=-\frac{2 (848\,d^3-952\,d^2+135\,d-2)}{(52\,d-5) (712\,d^2-102\,d+1)}<0
$$
for all $d\ge1$.  It follows that the $(3,3,6)$ entry of Table \ref{reduction B3} is $*$ if $d\le10$ and D if $d\ge11$.  The other 483 cases are analyzed similarly.
\end{proof}

Requiring that Banker make the optimal move in each of the cases that do not depend on Player's strategy, we have reduced the game, under Model B3 (resp., B2), to a $2^5\times 2^{n_d}$ (resp., $2\times 2^{n_d}$) matrix game, where $18\le n_d\le 23$.   

We have similar results for Model B1.  Again, $n_d:=|T_*|$.

\begin{theorem}\label{thm reduce B1}
\emph{(a)} Under Model B1 with the number of decks being a positive integer $d$, Lemma \ref{Lemma1} applies.  The sets $T_0$, $T_1$, and $T_*$ of the lemma can be inferred from Table \ref{reduction B1}, with entries \emph{S}, \emph{D}, and $*$ located at elements of $T_0$, $T_1$, and $T_*$, respectively.  In particular, $n_1=4$, $n_2=3$, and, $n_d=4$ for all $d\ge3$.

\emph{(b)} Lemma \ref{Lemma1} applies under Model A1, with the results those under Model B1 with $d\ge4$. 
\end{theorem}

\begin{table}[htb]
\caption{\label{reduction B1}Banker's optimal move (preliminary version) under Model B1 with $d=6$, indicated by D (draw) or S (stand), except in the $n_6=4$ cases indicated by $*$ in which it depends on Player's strategy.  Adjustments to the table for other positive integers $d$ are specified by footnotes.  Under Model A1, results are the same as those under Model B1 with $d\ge4$.\medskip}
\catcode`@=\active\def@{\phantom{0}}
\catcode`#=\active\def#{\phantom{$^0$}}
\begin{center}
\begin{tabular}{ccccccccccccc}
\hline
\noalign{\smallskip}
Banker's &\multicolumn{11}{c}{Player's third card ($\varnothing$ if Player stands)}\\
two-card &&&&&&&&&&\\
total & 0# & 1# & 2# & 3# & 4# & 5# & 6# & 7# & 8# & 9# & $\varnothing$ \\
\noalign{\smallskip} \hline
\noalign{\smallskip}
$0,1,2$& \cellcolor[gray]{0.85}D# & \cellcolor[gray]{0.85}D# & \cellcolor[gray]{0.85}D# & \cellcolor[gray]{0.85}D# & \cellcolor[gray]{0.85}D# & \cellcolor[gray]{0.85}D# & \cellcolor[gray]{0.85}D# & \cellcolor[gray]{0.85}D# & \cellcolor[gray]{0.85}D# & \cellcolor[gray]{0.85}D# & \cellcolor[gray]{0.85}D\\
3&       \cellcolor[gray]{0.85}D# & \cellcolor[gray]{0.85}D# & \cellcolor[gray]{0.85}D# & \cellcolor[gray]{0.85}D# & \cellcolor[gray]{0.85}D# & \cellcolor[gray]{0.85}D# & \cellcolor[gray]{0.85}D# & \cellcolor[gray]{0.85}D# & S$^3$ & $*$# & \cellcolor[gray]{0.85}D\\
4&       S# & $*^3$ & \cellcolor[gray]{0.85}D$^1$ & \cellcolor[gray]{0.85}D# & \cellcolor[gray]{0.85}D# & \cellcolor[gray]{0.85}D# & \cellcolor[gray]{0.85}D# & \cellcolor[gray]{0.85}D# & S# & S# & \cellcolor[gray]{0.85}D\\
5&       S# & S# & S# & S# & $*^2$ & \cellcolor[gray]{0.85}D# & \cellcolor[gray]{0.85}D# & \cellcolor[gray]{0.85}D# & S# & S# & \cellcolor[gray]{0.85}D\\
6&       S# & S# & S# & S# & S# & S# & \cellcolor[gray]{0.85}D# & \cellcolor[gray]{0.85}D# & S# & S# & $*$\\
7&       S# & S# & S# & S# & S# & S# & S# & S# & S# & S# & S\\
\noalign{\smallskip}
\hline
\noalign{\smallskip}
\multicolumn{12}{l}{$^1$Replace D by $*$ if $d=1$.}\\
\multicolumn{12}{l}{$^2$Replace $*$ by S if $d\le2$.}\\
\multicolumn{12}{l}{$^3$Replace S by $*$ and $*$ by S if $d\le3$.}\\
\end{tabular}
\end{center}
\end{table}

\begin{proof}
Let
$$
b_u(j,k):=a_{u,1}(j,k)-a_{u,0}(j,k),
$$
where $u=0$ corresponds to $S^\circ=\{0,1,2,3,4\}$ and $u=1$ corresponds to $S^\circ=\{0,1,2,3,4,5\}$.
If $b_u(j,k)<0$ for $u=0,1$, then the $(j,k)$ entry in Table \ref{reduction B1} is D; in 54 of the 88 entries, this property holds for every $d\ge1$.  If $b_u(j,k)>0$ for $u=0,1$, then the $(j,k)$ entry in Table \ref{reduction B1} is S; in 28 of the 88 entries, this property holds for every $d\ge1$.  If $b_0(j,k)>0>b_1(j,k)$ or $b_0(j,k)<0<b_1(j,k)$, then the $(j,k)$ entry in Table \ref{reduction B3} is $*$;  in two of the 88 entries, namely $(3,9)$ and $(6,\varnothing)$, this property holds for every $d\ge1$.  This accounts for all but four entries in Table \ref{reduction B1}, those marked with footnotes, in which the sign of $b_u(j,k)$ depends on $d$ for $u=0$ or $u=1$.   

For example, one of the four is $(5,4)$.  Indeed,
$$
b_1(5,4)=-\frac{15\,360\,d^4-45\,184\,d^3+9040\,d^2-588\,d+13}{(52\,d-5) (26\,880\,d^3-4680\,d^2+242\,d-3)}<0\text{ (resp., }>0)
$$
for $d\ge3$ (resp., $d\le2$), and
$$
b_0(5,4)=\frac{1024\,d^4+37\,248\,d^3-7792\,d^2+492\,d-7}{(52\,d-5)(22\,784\,d^3-3976\,d^2+194\,d-1)}>0
$$
for all $d\ge1$.  Therefore, the entry in the $(5,4)$ position of Table \ref{reduction B1} is $*$ if $d\ge3$ and S if $d\le2$.  The other 87 cases are analyzed similarly.
\end{proof}

\section{Solutions under Models A1, A2, and A3}\label{A1}

We recall Kemeny and Snell's (1957) solution of the $2\times 2^{88}$ matrix game that assumes Model A1.  (See Deloche and Oguer 2007 for an alternative approach based on the extensive, rather than the strategic, form of the game.)  Implicitly using Lemma~\ref{Lemma1}, they reduced the number of Banker pure strategies from $2^{88}$ to just $2^4$ (see Theorem \ref{thm reduce B1}).  The $2\times 2$ kernel of the game was determined to be
\begin{equation}\label{kernel A1}
\bordermatrix{&\text{B: S on }6,\varnothing & \text{B: D on }6,\varnothing\cr
\text{P: S on }5 &-4564 & -2692\cr
\text{P: D on }5 &-3705 & -4121\cr}2^4/(13)^6,
\end{equation}
implying that the following Player and Banker strategies are uniquely optimal.  Player draws on a two-card total of 5 with probability 
\begin{equation}\label{p A1}
p=\frac{9}{11}\approx0.818182,
\end{equation}
and Banker draws on a two-card total of $6$, when Player stands, with probability
\begin{equation}\label{q A1}
q=\frac{859}{2288}\approx0.375437.
\end{equation}
The value of the game (to Player) is 
\begin{equation}\label{v A1}
v=-\frac{679\,568}{53\,094\,899}\approx-0.0127991.
\end{equation}
The fully specified optimal strategy for Banker is given in Table \ref{banker B1} with $d\ge4$ and $q$ as in (\ref{q A1}).

Let us extend this analysis from Model A1 to Model A3.  Again we have a $2^5\times2^{484}$ matrix game, and the payoff matrix can be evaluated as in Section \ref{reduce}, using (\ref{case 1, d infinite}), (\ref{case 3, d infinite}), and (\ref{with-repl1})--(\ref{with-repl4}) in place of (\ref{case 1, d finite}), (\ref{case 3, d finite}), and (\ref{wo-repl1})--(\ref{wo-repl4}).  We can apply Lemma~\ref{Lemma1} and reduce the game to a $2^5\times 2^{22}$ matrix game.  We obtain the special case of Table~\ref{reduction B3} in which $d\ge11$.

When we evaluate this $2^5\times 2^{22}$ payoff matrix, we find that a number of rows are identical.  When several rows are identical, we eliminate duplicates.  When we make this reduction and rearrange the remaining rows in a more natural order, we are left with 9 rows, labeled by 0--8, which have a special structure.  Specifically, row $i$ corresponds to Player's mixed strategy (under Model A1) of drawing on a two-card total of 5 with probability $i/8$.  The reason that multiples of 1/8 appear is that, given that Player has a two-card total of 5, he has $(0,5)$, $(1,4)$, $(2,3)$, $(6,9)$, or $(7,8)$ with probabilities  $4/8$, $1/8$, $1/8$, $1/8$, and  $1/8$, respectively.

There are also a number of identical columns.  When we apply a similar reduction and rearrangement to the columns, we are left with $9^2(17)^2=23\,409$ columns, labeled by
$$
\{0,1,\ldots,8\}\times\{0,1,\ldots,16\}\times\{0,1,\ldots,8\}\times\{0,1,\ldots,16\},
$$
which have a similar structure.  Specifically, column $(j_1,j_2,j_3,j_4)$ corresponds to Banker's mixed strategy (under Model A1) of drawing on a two-card total of 3, when Player's third card is 9, with probability $j_1/8$; of drawing on a two-card total of 4, when Player's third card is 1, with probability $j_2/16$; of drawing on a two-card total of 5, when Player's third card is 4, with probability $j_3/8$; and of drawing on a two-card total of 6, when Player stands, with probability $j_4/16$.  The reason that multiples of 1/8 or 1/16 appear is that, given that Banker has a two-card total of 3, he has $(0,3)$, $(1,2)$, $(4,9)$, $(5,8)$, or $(6,7)$ with probabilities $4/8$, $1/8$, $1/8$, $1/8$, and $1/8$, respectively; given that Banker has a two-card total of 4, he has $(0,4)$, $(1,3)$, $(2,2)$, $(5,9)$, $(6,8)$, or $(7,7)$ with probabilities $8/16$, $2/16$, $1/16$, $2/16$, $2/16$, and $1/16$, respectively; given that Banker has a two-card total of 5, he has $(0,5)$, $(1,4)$, $(2,3)$, $(6,9)$, or $(7,8)$ with probabilities $4/8$, $1/8$, $1/8$, $1/8$, and $1/8$, respectively; given that Banker has a two-card total of 6, he has $(0,6)$, $(1,5)$, $(2,4)$, $(3,3)$, $(7,9)$, or $(8,8)$ with probabilities $8/16$, $2/16$, $2/16$, $1/16$, $2/16$, and $1/16$, respectively.

Next, we observe that column $(j_1,j_2,j_3,j_4)$ is a mixture of the $2^4$ (Model A1) pure strategies of Banker that remain after application of Lemma \ref{Lemma1}.  By the results for Model A1, optimal strategies for Banker must satisfy $j_1=8$, $j_2=0$, and $j_3=8$.  This reduces the game to a $9\times17$ matrix game, whose columns we relabel as 0--16.  Specifically, column $j$ corresponds to Banker's mixed strategy (under Model A1) of drawing on a two-card total of 6, when Player stands, with probability $j/16$.

Finally, what is the solution of the $9\times 17$ game?  We have seen that rows 1--7 (resp., columns 1--15) are mixtures of rows 0 and 8 (resp., columns 0 and 16).  In particular, rows 1--7 and columns 1--15 are dominated, but not strictly dominated.  Eliminating these rows and columns results in a $2\times2$ matrix game, namely the kernel (\ref{kernel A1}).  But eliminating dominated, but not strictly dominated, rows and columns may result in a loss of solutions, and it does so in this case.  Indeed, there are many solutions.  For Player, given two pure strategies $i,i'\in\{0,1,2,\ldots,8\}$ with
\begin{equation}\label{ineq3}
\frac{i}{8}<\frac{9}{11}<\frac{i'}{8},
\end{equation}
there is a unique $p\in(0,1)$ such that 
$$
(1-p)\bigg(\frac{i}{8}\bigg)+p\bigg(\frac{i'}{8}\bigg)=\frac{9}{11},
$$
and the $(1-p,p)$-mixture of pure strategies $i$ and $i'$ is optimal for Player.
There are 7 choices of $i$ and 2 choices of $i'$ that satisfy (\ref{ineq3}), hence 14 pairs $(i,i')$ that meet this condition.  

For Banker, given two pure strategies $j,j'\in\{0,1,2,\ldots,16\}$ with
\begin{equation}\label{ineq4}
\frac{j}{16}<\frac{859}{2288}<\frac{j'}{16},
\end{equation}
there is a unique $q\in(0,1)$ such that 
$$
(1-q)\bigg(\frac{j}{16}\bigg)+q\bigg(\frac{j'}{16}\bigg)=\frac{859}{2288},
$$
and the $(1-q,q)$-mixture of pure strategies $j$ and $j'$ is optimal for Banker.  There are 7 choices of $j$ and 10 choices of $j'$ that satisfy (\ref{ineq4}), hence 70 pairs $(j,j')$ that meet this condition.

Each such pair $(i,i')$ can be combined with each such pair $(j,j')$, so there are $14\times70=980$ pairs of optimal strategies of this form.  These are the extreme points of the convex set of equilibria.  All 980 of them appear when the game solver at \url{http://banach.lse.ac.uk/} is applied.

Let us single out one of them.  Take $i=6$ and $i'=7$, getting $p=6/11$, and take $j=1$ and $j'=9$, getting $q=179/286$.  This does not uniquely determine a pair of optimal strategies because of the duplicate rows and columns that were eliminated, but one pair of optimal mixed strategies to which this corresponds is shown in Table \ref{infdeck}.  
\begin{table}[htb]
\caption{\label{infdeck}A pair of optimal mixed strategies in Model A3.  For Banker's fully specified optimal strategy, see Table \ref{banker B3} with $d\ge10$.\medskip}
\catcode`@=\active\def@{\phantom{0}}
\begin{center}
\begin{tabular}{cc}\hline
\noalign{\smallskip}
\noalign{\smallskip}
\multicolumn{2}{l}{Player's two-card total is 5}\\
\noalign{\smallskip}
\hline
\noalign{\smallskip}
$(0,5),(6,9),(7,8)$ &  D   \\
$(1,4)$ &  $(\text{S},\text{D})$ with $(5/11,6/11)$   \\
$(2,3)$ &  S   \\
\noalign{\smallskip}
\hline
\noalign{\smallskip}
\noalign{\smallskip}
\multicolumn{2}{l}{Banker's two-card total is 6 and Player stands}\\
\noalign{\smallskip}
\hline
\noalign{\smallskip}
$(0,6)$ & $(\text{S},\text{D})$ with $(107/286,179/286)$\\
$(1,5),(2,4),(3,3),(7,9)$ &  S \\
$(8,8)$ &   D\\
\noalign{\smallskip}
\hline
\end{tabular}
\end{center}
\end{table}
As we will see, this pair of optimal mixed strategies is the limiting pair of optimal mixed strategies under Model B3 as $d\to\infty$.

Foster (1964) remarked, ``It is an interesting fact that this [optimal Player mixed] strategy is often attained approximately in practice by standing on the pair $2,3$ and calling [i.e., drawing] on any other combination adding to 5; this gives approximately the right frequency of calling [i.e., drawing].''   Actually, it gives a drawing probability of 7/8, not the required 9/11. But, as Table \ref{infdeck} suggests, Player should stand also on $(1,4)$ with probability 5/11.  Then the probability of Player drawing on a two-card total of 5 is
$$
\frac{3}{4}+\frac{1}{8}\,\frac{6}{11}=\frac{9}{11}.
$$

A similar analysis applies under Model A2.  Lemma \ref{Lemma1} reduces the $2\times 2^{484}$ matrix game to $2\times 2^{22}$.  Eliminating duplicate columns reduces the game to $2\times 9^2(17)^2$, and finally using the optimal solution under Model A1, we are left with a $2\times 17$ matrix game.  This is the game identified by Downton and Holder (1972).  As above, there are 70 extremal solutions, and they all appear when the game solver at \url{http://banach.lse.ac.uk/} is applied.

Again, we single out one of them.  Player draws on a two-card total of 5 with probability 9/11, and Banker follows Table \ref{infdeck}.  As we will see, this pair of optimal mixed strategies is the limiting pair of optimal mixed strategies under Model B2 as $d\to\infty$.

\section{Solution under Model B1}\label{B1}

Recall that Table \ref{reduction B1} applies under Model B1.  See Theorem \ref{thm reduce B1}.

In the case $d=6$, the $2\times2^4$ matrix game has kernel given by 
\begin{equation*}\label{kernel B1}
\bordermatrix{&\text{B: S on }6,\varnothing & \text{B: D on }6,\varnothing\cr
\text{P: S on 5 }&-23\,256\,431\,632& -13\,884\,629\,124\cr
\text{P: D on 5 }&-18\,880\,657\,128& -21\,061\,456\,188\cr}/1\,525\,814\,595\,305,
\end{equation*}
implying that the following Player and Banker strategies are uniquely optimal.  Player draws on a two-card total of 5 with probability 
\begin{equation}\label{p6 B1}
p_6=\frac{7\,631\,761}{9\,407\,656}\approx0.811229,
\end{equation}
and Banker draws on a two-card total of 6, when Player stands, with probability
\begin{equation}\label{q6 B1}
q_6=\frac{546\,971\,813}{1\,444\,075\,196}\approx0.378770.
\end{equation}
The value of the game (to Player) is 
\begin{equation}\label{v6 B1}
v_6=-\frac{23\,174\,205\,422\,119\,131}{1\,794\,292\,354\,051\,081\,885}\approx-0.0129155.
\end{equation}
The fully specified optimal strategy for Banker is given in Table \ref{banker B1}.  Comparing the solution under Model A1 with that under Model B1 reveals the effect of the ``with replacement'' assumption.  The solutions are identical except for the three parameters [(\ref{p A1})--(\ref{v A1}) vs.\ (\ref{p6 B1})--(\ref{v6 B1})].  For each parameter, the relative error is less than one percent.

\begin{table}[htb]
\caption{\label{banker B1}Banker's optimal move (final version) under Model B1 with $d$ being a positive integer, indicated by D (draw) or S (stand), or $(\text{S},\text{D})$ (stand with probability $1-q$, draw with probability $q$).  Here $q$ is as in (\ref{q1-3 B1})--(\ref{q4+ B1}).  Under Model A1, results are the same as those under Model B1 with $d\ge4$, except that $q$ is as in (\ref{q A1}).\medskip}
\catcode`@=\active\def@{\phantom{0}}
\catcode`#=\active\def#{\phantom{$^0$}}
\begin{center}
\begin{tabular}{ccccccccccccc}
\hline
\noalign{\smallskip}
Banker's &\multicolumn{11}{c}{Player's third card ($\varnothing$ if Player stands)}\\
two-card &&&&&&&&&&\\
total & 0# & 1# & 2# & 3# & 4# & 5# & 6# & 7# & 8# & 9# & $\varnothing$ \\
\noalign{\smallskip} \hline
\noalign{\smallskip}
$0,1,2$& \cellcolor[gray]{0.85}D# & \cellcolor[gray]{0.85}D# & \cellcolor[gray]{0.85}D# & \cellcolor[gray]{0.85}D# & \cellcolor[gray]{0.85}D# & \cellcolor[gray]{0.85}D# & \cellcolor[gray]{0.85}D# & \cellcolor[gray]{0.85}D# & \cellcolor[gray]{0.85}D# & \cellcolor[gray]{0.85}D# & \cellcolor[gray]{0.85}D\\
3&       \cellcolor[gray]{0.85}D# & \cellcolor[gray]{0.85}D# & \cellcolor[gray]{0.85}D# & \cellcolor[gray]{0.85}D# & \cellcolor[gray]{0.85}D# & \cellcolor[gray]{0.85}D# & \cellcolor[gray]{0.85}D# & \cellcolor[gray]{0.85}D# & S# & \cellcolor[gray]{0.85}D# & \cellcolor[gray]{0.85}D\\
4&       S# & S# & \cellcolor[gray]{0.85}D# & \cellcolor[gray]{0.85}D# & \cellcolor[gray]{0.85}D# & \cellcolor[gray]{0.85}D# & \cellcolor[gray]{0.85}D# & \cellcolor[gray]{0.85}D# & S# & S# & \cellcolor[gray]{0.85}D\\
5&       S# & S# & S# & S# & $*$# & \cellcolor[gray]{0.85}D# & \cellcolor[gray]{0.85}D# & \cellcolor[gray]{0.85}D# & S# & S# & \cellcolor[gray]{0.85}D\\
6&       S# & S# & S# & S# & S# & S# & \cellcolor[gray]{0.85}D# & \cellcolor[gray]{0.85}D# & S# & S# & \!\!\!\!\!$(\text{S},\text{D})$\!\!\!\!\\
7&       S# & S# & S# & S# & S# & S# & S# & S# & S# & S# & S\\
\noalign{\smallskip}
\hline
\noalign{\smallskip}
\multicolumn{12}{l}{$^*$ D if $d\ge4$, S if $d\le3$.}\\
\end{tabular}
\end{center}
\end{table}

This analysis extends to every positive integer $d$.

\begin{theorem}\label{thm opt B1}
Under Model B1 with $d$ being a positive integer, the following Player and Banker strategies are uniquely optimal.  Player draws on a two-card total of 5 with probability
\begin{equation}\label{pd B1}
p_d=\frac{(36\,864\,d^3-9312\,d^2+732\,d-23)}{8(5632\,d^3-1138\,d^2+69\,d-1)},\quad d\ge1,
\end{equation}
and Banker draws on on a two-card total of 6, when Player stands, with probability 
\begin{equation}\label{q1-3 B1}
q_d=\frac{224\,000\,d^4-55\,712\,d^3+2936\,d^2+163\,d-14}{2(52\,d-5)(5632\,d^3-1138\,d^2+69\,d-1)},\quad 1\le d\le3,
\end{equation}
or
\begin{equation}\label{q4+ B1}
q_d=\frac{439\,808\,d^4-107\,456\,d^3+5248\,d^2+374\,d-31}{4(52\,d-5)(5632\,d^3-1138\,d^2+69\,d-1)},\quad d\ge4.
\end{equation}
The value of the game (to Player) is
\begin{eqnarray*}
v_d&=&-32\,d^2(44\,396\,707\,840\,d^7-18\,908\,426\,240\,d^6+3\,279\,293\,696\,d^5\\
&&\qquad\quad{}-294\,129\,728\,d^4+14\,418\,160\,d^3-407\,352\,d^2+9543\,d\\
&&\qquad\quad{}-220)/[(5632\,d^3-1138\,d^2+69\,d-1)(52\,d)_6],\quad 1\le d\le3,
\end{eqnarray*}
or
\begin{eqnarray}\label{v4+ B1}
v_d&=&-16\,d^2(89\,072\,336\,896\,d^7-38\,873\,874\,432\,d^6+6\,969\,345\,536\,d^5\nonumber\\
&&\qquad\quad{}-655\,761\,920\,d^4+34\,638\,784\,d^3-1\,090\,952\,d^2+26\,286\,d\nonumber\\
&&\qquad\quad{}-537)/[(5632\,d^3-1138\,d^2+69\,d-1)(52\,d)_6],\quad d\ge4.\qquad
\end{eqnarray}
The fully specified optimal strategy for Banker is given in Table \ref{banker B1}.  
\end{theorem}

\begin{proof}
For $d\ge4$, the kernel of the $2\times2^4$ payoff matrix is given by columns 10 and 11 (when columns are labeled from 0 to 15), namely
\begin{equation*}\label{kernel B1 d}
\bordermatrix{&\text{B: S on }6,\varnothing & \text{B: D on }6,\varnothing\cr
\text{P: S on 5}&4(1\,168\,384\,d^4-284\,720\,d^3 &\!\!\! 2\,756\,608\,d^4-470\,336\,d^3\cr
&{}+22\,320\,d^2-446\,d-11)&\!\!\!\;{}+4656\,d^2+3072\,d-159\cr
\noalign{\medskip}
\text{P: D on 5}&2(1\,896\,960\,d^4-461\,984\,d^3 &\!\!\! 4\,219\,904\,d^4-954\,112\,d^3\cr
&{}+39\,392\,d^2-1266\,d+9)&\!\!\!\;{}+68\,384\,d^2-852\,d-57}\frac{(-64\,d^2)}{(52\,d)_6}.
\end{equation*}
This implies (\ref{pd B1}) for $d\ge4$, (\ref{q4+ B1}), and (\ref{v4+ B1}).  

To confirm this, we must show that, with $\bm A$ denoting the $2\times2^4$ payoff matrix, we have
$$
(1-p_d,p_d)\bm A\ge(v_d,v_d,\ldots,v_d).
$$
This involves checking 16 inequalities (of which two are automatic).  For example, the eighth and ninth components of $(1-p_d,p_d)\bm A-(v_d,v_d,\ldots,v_d)$ equal 
\begin{eqnarray*}
&&16\,d^2 (278\,921\,216\,d^7 - 1\,057\,021\,952\,d^6 + 410\,758\,144\,d^5 - 67\,502\,464\,d^4\\
&&\qquad{} + 5\,802\,464\,d^3 - 276\,248\,d^2 + 7200\,d -97)\\
&&\qquad\qquad /[(5632\,d^3 - 1138\,d^2 + 69\,d - 1)(52\,d)_6],
\end{eqnarray*}
which is positive for $d\ge4$ and negative for $1\le d\le3$.  The 10th and 11th components are 0, of course.  The remaining components are positive for every $d\ge1$.

A similar analysis can then be carried out for $1\le d\le3$, in which case the kernel is given by columns 8 and 9 (of 0--15).
\end{proof}

We notice that the above kernel converges, as $d\to\infty$, to the kernel (\ref{kernel A1}).

\section{Solution under Model B2}\label{B2} 

The result that Table \ref{reduction B3} is identical under Models B2 and B3 is less surprising than it may first appear to be.  In Section \ref{A1} we saw that, under Model A3 with Player's pure strategies labeled from 0 to 31, pure strategy $u\in\{1,2,\ldots,30\}$ is a $(1-p,p)$ mixture of pure strategies 0 and 31, where
\begin{equation}\label{mixing}
p=\frac{4u_1+u_2+u_3+u_4+u_5}{8}\in\Big\{\frac{1}{8},\frac{2}{8},\ldots,\frac{7}{8}\Big\};
\end{equation}
here $u_1u_2u_3u_4u_5$ is the binary form of $u$, that is, $u_1,u_2,u_3,u_4,u_5\in\{0,1\}$ and $u=16u_1+8u_2+4u_3+2u_4+u_5$.
Consequently,
\begin{equation}\label{between}
a_{u,l}(j_1,j_2,k)\text{ lies between }a_{0,l}(j_1,j_2,k)\text{ and }a_{31,l}(j_1,j_2,k)
\end{equation}
for all $u\in\{1,2,\ldots,30\}$, $l=0,1$, $0\le j_1\le j_2\le9$ with $M(j_1+j_2)\le7$, and $k=0,1,\ldots,9,\varnothing$.  
Under Model B3, the conditional expectations in (\ref{between}) should not differ much from their Model A3 counterparts, especially for large $d$, hence we would expect that (\ref{between}) holds with few exceptions.  In fact, the \textit{only} exceptions occur when $l=1$, $M(j_1+j_2)=2$, $k=8$, and $d\le7$ because in these cases, $a_{0,l}(j_1,j_2,k)$ and $a_{31,l}(j_1,j_2,k)$ are very close.  When we consider the differences $b_u(j_1,j_2,k)$, there are even fewer exceptions, as noted previously.

One might ask why Model B2 was even considered by Downton and Lockwood (1975), inasmuch as its asymmetric assumption about the available information (beyond the asymmetry inherent in the rules) may seem contrived.  The answer, we believe, is that there already existed an algorithm, due to Foster (1964), for solving such games.  That algorithm was formalized in Lemma \ref{Lemma2} of Section \ref{prelims}.

Let us recall Downton and Lockwood's (1975) solution of the $2\times 2^{484}$ matrix game that assumes Model B2.  In the case $d=6$, the $2\times2$ kernel of the game is found to be
\begin{equation*}\label{kernel B2}
\bordermatrix{&\text{B: S on }(0,6),\varnothing & \text{B: D on }(0,6),\varnothing\cr
\text{P: S on 5 }&-22\,721\,165\,499& -18\,033\,241\,115\cr
\text{P: D on 5 }&-19\,018\,265\,931& -20\,151\,297\,323\cr}/1\,525\,814\,595\,305,
\end{equation*}
implying that the following Player and Banker strategies are uniquely optimal.  Player draws on a two-card total of 5 with probability 
\begin{equation*}
p_6=\frac{477\,191}{592\,524}\approx0.805353,
\end{equation*}
and Banker draws on $(0,6)$, when Player stands, with probability
\begin{equation*}
q_6=\frac{77\,143\,741}{121\,269\,912}\approx0.636133.
\end{equation*}
The value of the game (to Player) is 
\begin{equation}\label{v6 B2}
v_6=-\frac{974\,653\,793\,197\,999}{75\,340\,147\,272\,374\,985}\approx-0.0129367,
\end{equation}
which is less than (\ref{v6 B1}) because Banker has additional options while Player's options are unchanged.  The fully specified optimal strategy for Banker is given in Table \ref{banker B3}.

This analysis extends to every positive integer $d$.  We do not display the kernel, only the solution.  

\begin{theorem}\label{thm opt B2}
Under Model B2 with $d$ being a positive integer, the following Player and Banker strategies are uniquely optimal.
Player draws on a two-card total of 5 with probability
\begin{equation*}
p_d=\frac{(8\,d-1)(12\,d-1)(24\,d-1)}{2\,d(1408\,d^2-220\,d+9)},\quad d\ge1,
\end{equation*}
and Banker draws on $(0,6)$, when Player stands, with probability 
\begin{equation}\label{q1-3 B2}
q_1=\frac{290\,383}{450\,072},\quad q_2=\frac{2\,591\,845}{4\,119\,192},\quad q_3=\frac{9\,294\,089}{14\,521\,368},
\end{equation}
\begin{equation}\label{q4-7 B2}
q_d=\frac{368\,640\,d^4-68\,624\,d^3-2168\,d^2+981\,d-48}{8\,d(52\,d-5)(1408\,d^2-220\,d+9)},\quad 4\le d\le7,
\end{equation}
\begin{equation}\label{q8,9 B2}
q_d=\frac{367\,616\,d^4-67\,728\,d^3-2416\,d^2+1015\,d-51}{8\,d(52\,d-5)(1408\,d^2-220\,d+9)},\quad d=8,9,
\end{equation}
or 
\begin{equation}\label{q10+ B2}
q_d=\frac{366\,592\,d^4-67\,344\,d^3-2456\,d^2+1017\,d-51}{8\,d(52\,d-5)(1408\,d^2-220\,d+9)},\quad d\ge10.
\end{equation}
The value of the game (to Player) is 
\begin{equation*}
v_1=-\frac{22\,932\,137}{1\,666\,583\,100},\quad v_2=-\frac{8\,220\,886\,553}{620\,866\,384\,425},\quad v_3=-\frac{210\,084\,639\,838}{16\,053\,072\,820\,785},
\end{equation*}
\begin{eqnarray*}
v_d&=&-32\,d(11\,125\,325\,824\,d^7-4\,182\,669\,312\,d^6+615\,333\,888\,d^5\nonumber\\
&&\qquad\quad{}-43\,467\,904\,d^4+1\,329\,008\,d^3+5040\,d^2\nonumber\\
&&\qquad\quad{}-1551\,d+39)/[(1408\,d^2-220\,d+9)(52\,d)_6],\quad 4\le d\le7,\quad
\end{eqnarray*}
\begin{eqnarray*}
v_d&=&-32\,d(11\,129\,683\,968\,d^7-4\,218\,739\,712\,d^6+635\,681\,024\,d^5\nonumber\\
&&\qquad\quad{}-47\,725\,760\,d^4+1\,738\,944\,d^3-14\,344\,d^2\nonumber\\
&&\qquad\quad{}-1093\,d+33)/[(1408\,d^2-220\,d+9)(52\,d)_6],\quad d=8,9,\quad
\end{eqnarray*}
or
\begin{eqnarray}\label{v10+ B2}
v_d&=&-32\,d(11\,134\,042\,112\,d^7-4\,259\,389\,440\,d^6+648\,152\,320\,d^5\nonumber\\
&&\qquad\quad{}-49\,007\,232\,d^4+1\,788\,256\,d^3-14\,816\,d^2\nonumber\\
&&\qquad\quad{}-1089\,d+33)/[(1408\,d^2-220\,d+9)(52\,d)_6],\quad d\ge10.\qquad
\end{eqnarray}
The fully specified optimal strategy for Banker is given in Table \ref{banker B3}.  
\end{theorem}

\begin{remark}
This is a slightly stronger statement than that of Downton and Lockwood (1975).\footnote{They assumed $d\le8$ and rounded results to four decimal places, but they also allowed $d=\frac{1}{2}$, which we do not.  Their Table 2(b), which graphically represents the dependence on $d$ in Banker's optimal strategy, contains three ambiguities.  Specifically, for $(j_1,j_2,k)=(6,8,1)$ and $d=2$, for $(j_1,j_2,k)=(2,3,4)$ and $d=8$, and for $(j_1,j_2,k)=(6,9,4)$ and $d=2$, it is uncertain whether D or S was intended.  We have confirmed that D, S, and S, respectively, were intended in these cases.  Furthermore, their table seems to suggest that D applies when $(j_1,j_2,k)=(2,3,4)$ and $d=9$, which is incorrect.}  The reason for having different formulas for $q_d$ and $v_d$ in the six cases $d=1$, $d=2$, $d=3$, $4\le d\le7$, $d=8,9$, and $d\ge10$ is that Banker's two pure strategies that determine the kernel do not vary if $4\le d\le7$, if $d=8,9$, or if $d\ge10$.  This is a consequence of Table~\ref{banker B3}, which comes largely from Tables~2(a) and 2(b) of Downton and Lockwood (1975).
\end{remark}

\begin{proof}
We apply Lemma \ref{Lemma2} for $d=1,2,\ldots,19$.  Only $(3,3,6)$ belongs to $T_{10}$ (if $d\le10$).  For each choice of $d$, the program runtime is about 15 seconds.  For $d\ge20$, the 22 points $p(l)$ ($l\in T_{01}\cup T_{10}$), relabeled as $p_1,p_2,\ldots,p_{22}$, at which Player's expected payoff is evaluated satisfy $p_{16}<p_{12}<p_{15}<p_{13}<p_{14}<p_5<p_4<p_1<p_3<p_2<p_{22}<p_{17}<p_{21}<p_{18}<p_{19}<p_{20}<p_{11}<p_{10}<p_8<p_7<p_9<p_6$, so the algorithm applies with a variable $d$.  For $d\ge20$, $V(p)$ is maximized at $p_{17}$ (or $p(0,6,\varnothing)$).  At $p_{22}$ (or $p(8,8,\varnothing)$), for example, 
\begin{eqnarray*}
V(p_{22})&=&-32 d(2\,783\,510\,528\,d^7-1\,188\,571\,136\,d^6+203\,128\,704\,d^5\\
&&\qquad{}-16\,568\,896\,d^4+\,596\,408\,d^3-16\,158\,d^2\\
&&\qquad{}+1855\,d-93)/[(352\,d^2-71\,d+4)(52\,d)_6],
\end{eqnarray*}
and this is less than $V(p_{17})$ (see (\ref{v10+ B2})) for all $d\ge20$, though the difference tends to 0 as $d\to\infty$.
\end{proof}

\begin{table}[htb]
\caption{\label{banker B3}Banker's optimal move (final version) under Model B2 or B3 with $d$ being a positive integer, indicated by D (draw), S (stand), or $(\text{S},\text{D})$ (stand with probability $1-q$, draw with probability $q$).  In Model B2, $q$ is as in (\ref{q1-3 B2})--(\ref{q10+ B2}), and in Model B3, $q$ is as in (\ref{q1-3 B3})--(\ref{q9+ B3}).}
\catcode`@=\active\def@{\phantom{0}}
\catcode`#=\active\def#{\phantom{$^0$}}
\begin{center}
\begin{tabular}{ccccccccccccc}
\hline
\noalign{\smallskip}
Banker's &\multicolumn{11}{c}{Player's third card ($\varnothing$ if Player stands)}\\
two-card &&&&&&&&&&\\
total & 0# & 1# & 2# & 3# & 4# & 5# & 6# & 7# & 8# & 9# & $\varnothing$ \\
\noalign{\smallskip} \hline
\noalign{\smallskip}
$0,1,2$& \cellcolor[gray]{0.85}D# & \cellcolor[gray]{0.85}D# & \cellcolor[gray]{0.85}D# & \cellcolor[gray]{0.85}D# & \cellcolor[gray]{0.85}D# & \cellcolor[gray]{0.85}D# & \cellcolor[gray]{0.85}D# & \cellcolor[gray]{0.85}D# & \cellcolor[gray]{0.85}D# & \cellcolor[gray]{0.85}D# & \cellcolor[gray]{0.85}D\\
3&       \cellcolor[gray]{0.85}D# & \cellcolor[gray]{0.85}D# & \cellcolor[gray]{0.85}D# & \cellcolor[gray]{0.85}D# & \cellcolor[gray]{0.85}D# & \cellcolor[gray]{0.85}D# & \cellcolor[gray]{0.85}D# & \cellcolor[gray]{0.85}D# & S$^1$ & \cellcolor[gray]{0.85}D# & \cellcolor[gray]{0.85}D\\
4&       S# & S$^2$ & \cellcolor[gray]{0.85}D# & \cellcolor[gray]{0.85}D# & \cellcolor[gray]{0.85}D# & \cellcolor[gray]{0.85}D# & \cellcolor[gray]{0.85}D# & \cellcolor[gray]{0.85}D# & S# & S# & \cellcolor[gray]{0.85}D\\
5&       S# & S# & S# & S# & $*$# & \cellcolor[gray]{0.85}D# & \cellcolor[gray]{0.85}D# & \cellcolor[gray]{0.85}D# & S# & S# & \cellcolor[gray]{0.85}D\\
6&       S# & S# & S# & S# & S# & S# & \cellcolor[gray]{0.85}D$^3$ & \cellcolor[gray]{0.85}D# & S# & S# & $\dag$\\
7&       S# & S# & S# & S# & S# & S# & S# & S# & S# & S# & S\\
\noalign{\smallskip}
\hline
\end{tabular}
\begin{tabular}{cc}
\noalign{\smallskip}
\multicolumn{2}{l}{$^1$Banker's two-card total is 3 and Player's third card is 8}\\
\noalign{\smallskip}
$(0,3),(1,2),(5,8)$ &  S   \\
$(4,9)$ &  S if $d\ge2$, D if $d=1$   \\
$(6,7)$ & S if $d\ge1$ (Model B2)\\
 & S if $d\ge2$, D if $d=1$ (Model B3)\\
\noalign{\smallskip}
\hline
\noalign{\smallskip}
\multicolumn{2}{l}{$^2$Banker's two-card total is 4 and Player's third card is 1}\\
\noalign{\smallskip}
$(0,4),(1,3),(2,2),(5,9)$ &  S   \\
$(6,8),(7,7)$ &  S if $d\ge3$, D if $d\le2$   \\
\noalign{\smallskip}
\hline
\noalign{\smallskip}
\multicolumn{2}{l}{$^*$Banker's two-card total is 5 and Player's third card is 4}\\
\noalign{\smallskip}
$(0,5),(7,8)$ &  D if $d\ge2$, S if $d=1$   \\
$(1,4)$ &  S if $d\le7$, D if $d\ge8$  \\
$(2,3)$ &  S if $d\le9$, D if $d\ge10$ (Model B2)  \\
        &  S if $d\le8$, D if $d\ge9$ (Model B3)  \\
$(6,9)$ &  D if $d\ge3$, S if $d\le2$ (Model B2)\\
        &  D if $d\ge2$, S if $d=1$ (Model B3)\\
\noalign{\smallskip}
\hline
\noalign{\smallskip}
\multicolumn{2}{l}{$^3$Banker's two-card total is 6 and Player's third card is 6}\\
\noalign{\smallskip}
$(0,6),(1,5),(2,4),(7,9),(8,8)$ &  D   \\
$(3,3)$ &  D if $d\ge4$, S if $d\le3$   \\
\noalign{\smallskip}
\hline
\noalign{\smallskip}
\multicolumn{2}{l}{$^\dag$Banker's two-card total is 6 and Player stands} \\
\noalign{\smallskip}
$(0,6)$ & $(\text{S},\text{D})$ (Model B2) \\
 & $(\text{S},\text{D})$ if $d\ge2$, S if $d=1$ (Model B3) \\
$(1,5),(2,4),(3,3),(7,9)$ &    S \\
$(8,8)$ &   D (Model B2) \\
 & D if $d\ge2$, $(\text{S},\text{D})$ if $d=1$ (Model B3) \\
\noalign{\smallskip}
\hline
\end{tabular}
\end{center}
\end{table}
\afterpage{\clearpage}

\section{Solution under Model B3}\label{B3}

We have reduced the game to a $2^5\times2^{n_d}$ matrix game, where $18\le n_d\le23$.  The fact that Table \ref{reduction B3} is identical under Models B2 and B3 suggests the existence of a $2\times2$ kernel under Model B3 whose columns are the same as those of the $2\times2$ kernel under Model B2 as described in Section \ref{B2}.  Then the resulting $2^5\times2$ matrix game will of course have a $2\times2$ kernel, which is easy to find by graphical or other methods, and it will remain to  confirm that this kernel corresponds to a solution of the $2^5\times2^{n_d}$ matrix game.  

As we will see, this approach works for all positive integers $d$ except 1, 2, and 9.  These last three cases can be treated separately.

Let us begin with the case $d=6$.  Here the rows of the $32\times2$ payoff matrix are labeled from 0 to 31 as in Section \ref{reduce}, with the 5-bit binary form of the row number specifying the strategy (1 indicates D and 0 indicates S, in the five cases $(0,5)$, $(1,4)$, $(2,3)$, $(6,9)$, and $(7,8)$).  For example, row 19 (binary 10011) corresponds to drawing on $(0,5)$, standing on $(1,4)$ and $(2,3)$, and drawing on $(6,9)$ and $(7,8)$. 

We could label the two columns in a similar way but with binary strings of length 18 corresponding to the asterisks in Table \ref{reduction B3}, in the specific order $(0,3,9)$, $(1,2,9)$, $(4,9,9)$, $(5,8,9)$, $(6,7,9)$, $(2,2,1)$, $(6,8,1)$, $(7,7,1)$, $(0,5,4)$, $(6,9,4)$, $(7,8,4)$, $(3,3,6)$, $(0,6,\varnothing)$, $(1,5,\varnothing)$, $(2,4,\varnothing)$, $(3,3,\varnothing)$, $(7,9,\varnothing)$, $(8,8,\varnothing)$,  reading the string left to right.  In that case, the two columns would be labeled
\begin{equation}\label{binary}
111\,110\,001\,111\,000\,001\quad\text{and}\quad111\,110\,001\,111\,100\,001.
\end{equation}

The kernel is easily found to be given by rows 19 and 27, so it is equal to
\begin{equation*}
\bordermatrix{&\text{B: S on }(0,6),\varnothing & \text{B: D on }(0,6),\varnothing\cr
\text{P: S on }(1,4)&-19\,769\,569\,403 &-19\,425\,699\,931\cr 
\text{P: D on }(1,4)&-19\,391\,857\,983 &-19\,783\,609\,631\cr}/1\,525\,814\,595\,305,
\end{equation*}
implying that the following Player and Banker strategies are optimal.  Player draws on $(1,4)$ with probability 
\begin{equation}\label{p6 B3}
p_6=\frac{35\,003}{74\,880}\approx0.467455,
\end{equation}
and Banker draws on $(0,6)$, when Player stands, with probability
\begin{equation}\label{q6 B3}
q_6=\frac{18\,885\,571}{36\,781\,056}\approx0.513459.
\end{equation}
The value of the game (to Player) is 
\begin{equation}\label{v6 B3}
v_6=-\frac{73\,356\,216\,203\,119}{5\,712\,649\,844\,821\,920}\approx-0.0128410,
\end{equation}
which is greater than (\ref{v6 B2}) because Player has additional options while Banker's options are unchanged.
The fully specified optimal strategies for Player and Banker are given in Tables \ref{player} and \ref{banker B3}.

\begin{table}[tbh]
\caption{\label{player}Player's optimal move under Model B3, indicated by D (draw), S (stand), or $(\text{S},\text{D})$ (stand with probability $1-p$, draw with probability $p$).  Here $p$ is as in (\ref{pd B3}).
\medskip}
\catcode`@=\active\def@{\phantom{0}}
\begin{center}
\begin{tabular}{cc}\hline
\noalign{\smallskip}
Player's two-card hand & optimal move \\
\noalign{\smallskip} \hline
\noalign{\smallskip}
$(0,5),(6,9),(7,8)$& D \\
$(1,4)$& $(\text{S},\text{D})$\\
$(2,3)$&   S\\
\noalign{\smallskip}
\hline
\end{tabular}
\end{center}
\end{table}

This analysis extends to every positive integer $d$.  We do not display the kernel, only the solution.  

\begin{theorem}\label{thm opt B3}
Under Model B3 with the number of decks being a positive integer $d$, the following Player and Banker strategies are optimal.
Player draws on $(1,4)$ with probability
\begin{equation}\label{pd B3}
p_1=\frac{1}{19},\qquad  
p_d=\frac{(12\,d-1)(16\,d^2-14\,d+1)}{32\,d^2(11\,d-1)},\quad d\ge2,
\end{equation}
and Banker draws on
\begin{equation*}
\begin{cases}(8,8)&\text{if $d=1$,}\\(0,6)&\text{if $d\ge2$,}\end{cases}
\end{equation*}
when Player stands, with probability 
\begin{equation}\label{q1-3 B3}
q_1=\frac{4519}{10\,716},\quad q_2=\frac{17\,431}{64\,512},\quad q_3=\frac{4\,425\,647}{11\,132\,928},
\end{equation}
\begin{equation}\label{q4-7 B3}
q_d=\frac{92\,160\,d^4-120\,128\,d^3+26\,336\,d^2-2000\,d+47}{256\,d^2(11\,d-1)(52\,d-5)},\quad 4\le d\le7,
\end{equation}
\begin{equation}\label{q8 B3}
q_8=\frac{316\,815\,305}{585\,842\,688},
\end{equation}
or 
\begin{equation}\label{q9+ B3}
q_d=\frac{91\,648\,d^4-119\,488\,d^3+26\,032\,d^2-1932\,d+41}{256\,d^2(11\,d-1)(52\,d-5)},\quad d\ge9.
\end{equation}
The value of the game (to Player) is 
\begin{equation*}
v_1=-\frac{3\,439\,451}{25\,482\,800},\quad v_2=-\frac{49\,424\,010\,137}{3\,823\,801\,581\,600},\quad v_3=-\frac{31\,717\,439\,249}{2\,461\,444\,457\,472},
\end{equation*}
\begin{eqnarray*}
v_d&=&-2(1\,390\,665\,728\,d^7-491\,115\,520\,d^6+50\,698\,240\,d^5\nonumber\\
&&\qquad\quad{}+2\,428\,032\,d^4-990\,512\,d^3+89\,192\,d^2\nonumber\\
&&\qquad\quad{}-3462\,d+47)/[(11\,d-1)(52\,d)_6],\quad 4\le d\le7,\quad
\end{eqnarray*}
\begin{equation*}
v_8=-\frac{2\,789\,416\,947\,665\,657}{217\,430\,324\,984\,396\,160},
\end{equation*}
or
\begin{eqnarray}\label{vd B3}
v_d&=&-2(1\,391\,755\,264\,d^7-500\,535\,296\,d^6+54\,174\,464\,d^5\nonumber\\
&&\qquad\quad{}+1\,931\,136\,d^4-948\,816\,d^3+85\,792\,d^2\nonumber\\
&&\qquad\quad{}-3238\,d+41)/[(11\,d-1)(52\,d)_6],\quad d\ge9.\qquad
\end{eqnarray}
The fully specified optimal strategies for Player and Banker are given in Tables \ref{player} and \ref{banker B3}.
\end{theorem}

\begin{remark}
We suspect that the solution is unique but do not have a proof.  Notice what happens as $d\to\infty$:  $p_d\to6/11$, $q_d\to179/286$, and $v_d\to-679\,568/53\,094\,899$, so the optimal mixed strategies approach those of Table~\ref{infdeck}.
\end{remark}

\begin{proof}
Consider the case $d=6$.  Let $\bm A$ denote the $2^5\times2^{18}$ payoff matrix, and let $\bm p=(p^0,p^1,\ldots,p^{31})$ and $\bm q=(q^0,q^1,\ldots,q^{262\,143})$ denote the stated optimal strategies.  Specifically, with $p_6$ as in (\ref{p6 B3}) and $q_6$ as in (\ref{q6 B3}),
$$
p^{27}=1-p^{19}=p_6\quad\text{and}\quad p^i=0\text{ if }i\neq19,27.
$$
and
$$
q^{254\,945}=1-q^{254\,913}=q_6\quad\text{and}\quad q^j=0\text{ if }j\neq254\,913, 254\,945.
$$
($254\,913$ and $254\,945$ are the decimal forms of the binary numbers in (\ref{binary}).)  Then, with $v_6$ as in (\ref{v6 B3}), it suffices to show that
\begin{equation}\label{ineq1 B3}
\bm p\bm A\ge(v_6,v_6,\ldots,v_6)
\end{equation}
componentwise ($2^{18}$ inequalities), and
\begin{equation}\label{ineq2 B3}
\bm A\bm q^\textsf{T}\le(v_6,v_6,\ldots,v_6)^\textsf{T}
\end{equation}
componentwise ($2^5$ inequalities).  Note that it is not necessary to evaluate the $2^{23}$ entries of $\bm A$.  (\ref{ineq1 B3}) involves only rows 19 and 27 of $\bm A$ (when rows are labeled 0 to 31), while (\ref{ineq2 B3}) involves only columns $254\,913$ and $254\,945$ of $\bm A$ (when columns are labeled 0 to $262\,143$).  We have confirmed (\ref{ineq1 B3}) and (\ref{ineq2 B3}) using a \textit{Mathematica} program.  The program is, however, unnecessarily time-consuming.  A more efficient way to proceed is to use Lemma \ref{Lemma2} to verify (\ref{ineq1 B3}), with rows 19 and 27 (of 0--31) being Player's two pure strategies.  (In that case the number of exceptional cases is 6, namely $(0,6,\varnothing),(1,5,\varnothing),(2,4,\varnothing),(3,3,\varnothing),(7,9,\varnothing),(8,8,\varnothing)$, not 18.  $T_{10}$ is empty.)

Similar programs give analogous results for every positive integer $d$.  When $d$ is 1, 2, or 9 and we use the two Banker pure strategies whose mixture is optimal under Model B2, we find that (\ref{ineq1 B3}) fails.  By determining which components of the vector inequality fail, we can propose and confirm the correct optimal strategies under Model B3.  Specifically, if $d=2$ or $d=9$, then (\ref{ineq1 B3}) fails at just two components, namely the two that determine Banker's optimal mixed strategy under Model B3.  If $d=1$, then (\ref{ineq1 B3}) fails at eight components, which include the two optimal ones.  Some trial and error may be required in this case.

We hasten to add that, just as in the case $d=6$, there is a more efficient way.  The two Player pure strategies specified by the kernel do not vary with $d$ and are rows 19 and 27 (of 0--31).  We again apply Lemma \ref{Lemma2} to establish (\ref{ineq1 B3}), saving time and avoiding the issue that occurred in the cases $d=1,2,9$.  We find that, for $d=1,2,\ldots,9$, $n_d=7,9,8,9,6,6,6,7,7$, respectively, and, if $d\ge10$, $n_d=6$.  In each case, $T_{10}$ is empty.

More importantly, this last method allows us to treat the cases $d\ge10$ simultaneously, similarly to what we did in Section \ref{B2}.  For $d\ge10$, the six points $p(l)$ ($l\in\{(0,6,\varnothing),(1,5,\varnothing),(2,4,\varnothing),(3,3,\varnothing),(7,9,\varnothing),(8,8,\varnothing)\}$), relabeled as $p_1,p_2,\ldots,p_6$, at which Player's expected payoff is evaluated, satisfy $p_6<p_1<p_5<p_3<p_2<p_4$, so the algorithm applies with variable $d$.  For $d\ge10$, $V(p)$ is maximized at $p_1$ (or $p(0,6,\varnothing)$).  At $p_6$ (or $p(8,8,\varnothing)$), for example,
\begin{eqnarray*}
V(p_6)&=&-8(695\,877\,632\,d^7-281\,198\,592\,d^6+34\,472\,064\,d^5+1\,177\,024\,d^4\\
&&\quad{}-901\,592\,d^3+119\,896\,d^2-6755\,d+123)/[(22\,d-3)(52\,d)_6],
\end{eqnarray*}
and this is less than $V(p_1)$ (see (\ref{vd B3})) for all $d\ge10$, though the difference tends to 0 as $d\to\infty$.
\end{proof}

\end{document}